\documentclass[a4paper,onecolumn,10pt,accepted=2020-02-20]{quantumarticle}
\pdfoutput=1

\usepackage[utf8]{inputenc}
\usepackage{amsmath}
\usepackage{amsthm}
\usepackage{amssymb}
\usepackage{amsfonts}
\usepackage{dsfont} 
\usepackage{graphicx}
\usepackage{latexsym}
\usepackage{mathrsfs}

\usepackage{soul}
\usepackage{color}
\usepackage{xcolor}
\usepackage[colorlinks=true]{hyperref}
\hypersetup{
	citecolor={green!50!black},
	linkcolor={blue!70!black},
	urlcolor={blue!70!black}}

\newtheorem{definition}{Definition}
\newtheorem{proposition}{Proposition}
\newtheorem{lemma}{Lemma}

\newcommand{\Tr}{\mathrm{Tr}}
\newcommand{\Id}{\mathds{1}}
\newcommand{\ket}[1]{\left|{#1}\right\rangle}
\newcommand{\sket}[1]{|{#1}\rangle}
\newcommand{\bra}[1]{\left\langle{#1}\right|}
\newcommand{\sbra}[1]{\langle{#1}|}

\newcommand{\proj}[1]{\left|{#1}\right\rangle\!\!\left\langle{#1}\right|}
\newcommand{\sproj}[1]{|{#1}\rangle\!\langle{#1}|}
\newcommand{\norm}[1]{\left|\left|#1\right|\right|}

\newcommand{\bnorm}[1]{\big|\big|#1\big|\big|}

\newcommand{\cD}{\mathcal{D}}
\newcommand{\cE}{\mathcal{E}}
\newcommand{\cF}{\mathcal{F}}
\newcommand{\cM}{\mathcal{M}}

\newcommand{\oM}{\bar{M}}
\newcommand{\ocM}{\bar{\cM}}
\newcommand{\oocM}{\overrightarrow{\cM}^\cm}
\newcommand{\hil}{\mathcal{H}}

\renewcommand{\a}{\mathsf{a}}
\renewcommand{\b}{\mathsf{b}}
\newcommand{\f}{\mathsf{f}}
\newcommand{\x}{\mathsf{x}}
\newcommand{\y}{\mathsf{y}}
\newcommand{\X}{\mathsf{X}}
\newcommand{\A}{\mathsf{A}}
\newcommand{\vac}{\t{vac}}

\renewcommand{\t}[1]{\mathrm{#1}}
\newcommand{\tot}{\!\otimes\!}
\newcommand{\cm}{{\checkmark}}
\newcommand{\nc}{\textup{\,\o\,}}

\newcommand{\nn}{\nonumber}

\begin{document}

\title{How post-selection affects device-independent claims under the fair sampling assumption}

\author{Davide Orsucci}
\affiliation{Quantum Optics Theory Group, Universität Basel, CH-4056 Basel, Switzerland}

\author{Jean-Daniel Bancal}
\affiliation{Quantum Optics Theory Group, Universität Basel, CH-4056 Basel, Switzerland}
\affiliation{D\'epartement de Physique Appliqu\'ee, Universit\'e de Gen\`{e}ve, CH-1211 Gen\`{e}ve, Switzerland}

\author{Nicolas Sangouard}
\affiliation{Quantum Optics Theory Group, Universität Basel, CH-4056 Basel, Switzerland}

\author{Pavel Sekatski}
\affiliation{Quantum Optics Theory Group, Universität Basel, CH-4056 Basel, Switzerland}


\begin{abstract}

Device-independent certifications employ Bell tests to guarantee the proper functioning of an apparatus from the sole knowledge of observed measurement statistics, i.e.\ without assumptions on the internal functioning of the devices. When these Bell tests are implemented with devices having too low efficiency, one has to post-select the events that lead to successful detections and thus rely on a fair sampling assumption. The question that we address in this paper is what remains of a device-independent certification under fair sampling. We provide an intuitive description of post-selections in terms of \emph{filters} and define the fair sampling assumption as a property of these filters, equivalent to the definition introduced in Ref.~\cite{Berry10}. When this assumption is fulfilled, the post-selected data is reproduced by an ideal experiment where lossless devices measure a \emph{filtered} state which can be obtained from the \emph{actual} state via local probabilistic maps. Trusted conclusions can thus be obtained on the quantum properties of this filtered state and the corresponding measurement statistics can reliably be used, e.g., for randomness generation or quantum key distribution. We also explore a stronger notion of fair sampling leading to the conclusion that the post-selected data is a fair representation of the data that would be obtained with lossless detections. Furthermore, we show that our conclusions hold in cases of small deviations from exact fair sampling. Finally, we describe setups previously or potentially used in Bell-type experiments under fair sampling and identify the underlying device-specific assumptions.

\end{abstract}

\maketitle

\section*{Introduction}
\label{sec:intro}

Measurement devices in many quantum experiments, and notably in quantum optics, have a finite non-unit efficiency: such devices may refuse to provide an outcome for the desired measurement, and produce a ``no-click'' event instead~\cite{Had09}. This lack of detection can usually be explained in simple physical terms, e.g., the photon to be measured is not always absorbed on the chip of the detector. In a simplified physical model of the apparatus, the occurrence of a no-click event is completely independent of the state of the quantum system to be measured. Based on this model, one is tempted to ignore the no-click events altogether and remove them from the measurement data, a.k.a.\ performing a \emph{post-selection}.

However, from a black-box perspective, a no-click event must be considered as a measurement outcome, just like the others\footnote{A non-detection event is obtained when a detector is supposed to click but, instead, it doesn't. These non-detections can be distinguished from the normal ``silent'' state of the device by means of a \textit{trigger event}: a classical signal indicating that a quantum state has been sent to the detector within a certain time window.}. 
This outcome has to be added to the alphabet corresponding to possible values of the measured quantity~\cite{CHSH69, CH74, Pearle70}. The action of ignoring these events is not always harmless: in the context of testing of Bell inequalities this is known as the \textit{detection loophole}, and there are explicit examples where the use of post-selection can lead to wrong claims about the performed Bell test, e.g.\ erroneously deducing the non-locality of a local model~\cite{Pearle70, GM87, MP03}. Such canny local models are not abstract theoretical constructions, but have been produced experimentally with simple optical elements and detectors that are commonly used~\cite{PS11, GL11, RG13, JE15, Jogenfors17}.

Nevertheless, an accurate physical understanding of the measurement device \textit{in the lab} may suggest that the no-click events are locally random and independent of system state and measurement settings. If this is really the case, data acquisition and post-selection is equivalent to directly generating the data with an \textit{ideal} measurement device that has unit efficiency but is otherwise the same as the detector \textit{in the lab}. Thus, any claim that could be made in the ideal setting carries over when using the experimental post-selected data. This logical step is known as \textit{fair sampling assumption} (or \emph{no-enhancement assumption}, at the origins~\cite{CHSH69, CH74}). It is important to realize that fair sampling is indeed an assumption and cannot be ensured from the measurement probabilities alone -- post-selected local models as mentioned above can reproduce those perfectly.

Thanks to the advancement in quantum technologies we have recently witnessed the first experimental violations of Bell inequalities free of loopholes~\cite{NoLoop1, NoLoop2, NoLoop3, NoLoop4, NoLoop5}, conclusively ruling out that locally causal models could explain the correlations observed in quantum experiments. In the aftermath of achieving this long-standing goal, renewed interest in Bell inequalities was spurred by their technological applications within the framework of \emph{device-independent quantum information processing}~\cite{MY03, SB19}. Namely, the sole violation of a Bell inequality can be used to certify devices, to guarantee the randomness of measurement results or the security of a quantum key distribution, without the need to know the internal functioning of the devices used to perform the Bell test~\cite{PABG07, PABG09, MYS12, Kan16, SBWS18}. Device-independent quantum information processing is at the verge of experimental feasibility~\cite{Ban18}. However, performing Bell tests devoid of the detection loophole is still challenging and, as such, most experiments at present still rely on fair sampling. This naturally raises the following question: if assessing fair sampling requires a detailed description of the way the measurement device works, what remains of the device-independent framework in scenarios based on Bell tests relying upon the fair sampling assumption? The aim of our work is to provide a concrete answer to this question.

We first explain how a lossy detector can be understood in terms of \textit{filters}: we represent a finite-efficiency measurement as a two-step process where first a filter is applied to the classical input and to the quantum state which either accepts or rejects them (corresponding, respectively, to a successful or failed detection) and subsequently a lossless measurement is performed on the filtered quantum state. This is schematically illustrated in panel (a) of Figure~\ref{Fig1}. An equivalent formulation based on the positive operator-valued measure (POVM) description of the measurement devices is also provided; this formulation allows to directly verify if a detector satisfies fair sampling, given a full specification of its behaviour.

\begin{figure}[t]
\begin{center}
	\includegraphics[scale=1.2]{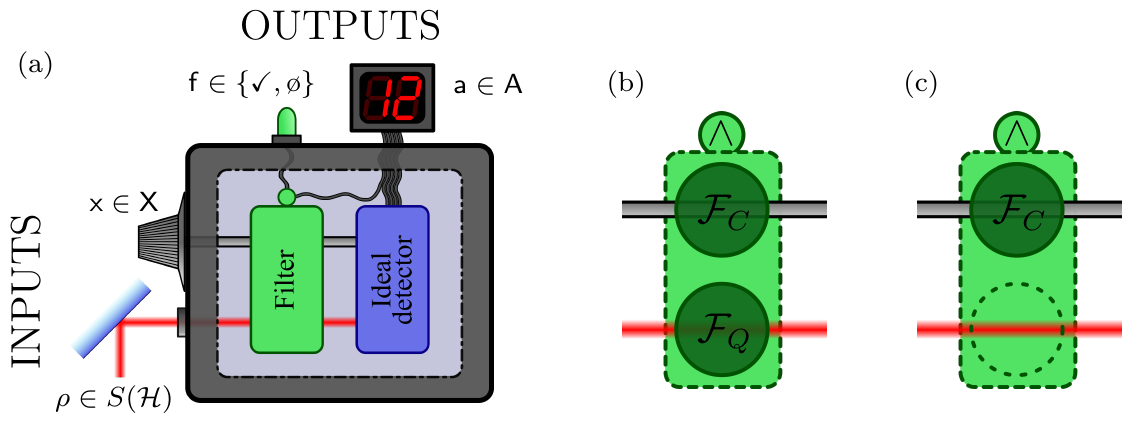}
	\vspace{-3mm}
\end{center}
\caption{Mathematical model of a lossy device. In panel (a) we illustrate that any device having finite efficiency can always be modelled as a \textit{filter} acting jointly on the classical setting $\x \in \X$ and on the input quantum state $\rho \in S(\hil)$, followed by  by an \textit{ideal lossless measurement device}. The filter returns a flag $\f \in \{\cm, \nc\}$ with a probability that depends on $\x$ and $\rho$, and the ideal measurement is performed only if $\f = \cm$. In panel (b) we illustrate the \emph{fair sampling} condition: the filter factorizes in two components $\cF_C$ and $\cF_Q$ which act independently on the classical and quantum input, respectively, and a detection occurs only when both accept (indicated by a $\wedge \equiv$ logical AND). In panel (c) we illustrate \emph{strong fair sampling}: the component acting on the quantum input is proportional to the identity channel.}
\label{Fig1}
\end{figure}

Following the work of Berry et.\ al.~\cite{Berry10}, we then give a very general formulation of fair sampling: to have fair sampling, it is sufficient to require that the filter operates independently on the classical input and on the quantum state, see panel (b) of Figure~\ref{Fig1}. Under this assumption the sampling can be regarded as ``fair'', since the post-selected data can be seen as being generated from an ideal quantum experiment where lossless devices measure a \emph{filtered} state\footnote{For instance, in quantum optics experiments the vacuum component of photonic states is often ignored. This corresponds to applying a filtering operation that projects the quantum state onto the non-vacuum component.}. This allows to reach trusted conclusions on the quantum properties of this filtered state and the corresponding measurement statistics can be used reliably for verifiable randomness generation~\cite{CK11, AM16} or for quantum key distribution~\cite{PABG07, PABG09}.

In this work, we also discuss a stronger version of the fair sampling assumption, corresponding to what is usually assumed to be the very definition of fair sampling (see e.g.~\cite{WJS98, MMB04}). In strong fair sampling, filter operating on the state is proportional to the identity, see panel (c) of Figure~\ref{Fig1}. We show that, in this case, the post-selected statistics is an unbiased representation of the statistics that would be obtained with unit-efficiency detectors measuring the actual \emph{experimental} state.

The motivation of the present work is thus to bridge the gap between theoretical studies of device-independent certifications and experimental realizations thereof. To this end, we describe concrete examples of experiments that can be used for device-independent quantum information processing and show that, under fair sampling, these experiments allow one to get semi-device-independent conclusions. In fact, even present-day experiments frequently rely upon a fair sampling assumption~\cite{exp1,exp2,exp3,exp4}, and this will likely continue to be necessary in future years; however, often there is very little or no theoretical analysis of the consequences that post-selection has in these experiments. Furthermore, we study what happens when small deviations from exact fair sampling are present and we show that the results are robust against small perturbations. Concretely, we provide a prescription for the minimal information that in an actual quantum experiment should be directly assessed in order to verify that a device (approximately) satisfies fair sampling. Assuming that this minimal information is publicly available and trusted, the certification of the device can be safely completed using a standard device-independent protocol with post-selected data.

\subsubsection*{Paper structure}

In Section~\ref{sec:preliminaries} we describe the general setup and establish part of the notation. In Section~\ref{sec:quantum_model} we give the general mathematical formulation of lossy measurement devices, as represented in panel~(a) of Figure~\ref{Fig1}. In Section~\ref{sec:fair_sampling} we formalize the definition of fair sampling as depicted in panel~(b) of Figure~\ref{Fig1} (equivalent to the definition introduced in reference~\cite{Berry10}); we then prove that, under this assumption, the post-selected data can be reproduced by an ideal quantum experiment involving only lossless detectors. We then introduce the stronger notion of fair sampling in Section~\ref{sec:strong_fs}, corresponding to panel~(c) of Figure~\ref{Fig1}. In Section~\ref{sec:crypto} we discuss the use of fair sampling in cryptographic settings. In Section~\ref{sec:fs_in_quantum_optics} we give concrete examples of fair sampling based on a quantum optics setup used to measure the polarization of photons. In Section~\ref{sec:approximate_fair sampling} we investigate the robustness of the results against small perturbations from exact fair sampling, and we show that the deviations in the post-selected data are linear in the perturbations from the exact case. We go back to the optical setup in Section~\ref{sec:appr_fs_in_quantum_optics} and analyse the consequence of approximate fair sampling in this context. Finally, in Section~\ref{sec:state-dependent} we show that even more general notions of fair sampling are possible, considering cases of fair sampling that are state-dependent. Supplementary material and further considerations are presented in the Appendices.

\section{Preliminary notions}
\label{sec:preliminaries}

\paragraph*{Bell tests:}
We consider a Bell test~\cite{Brunner14} involving two or more parties, which we will denote as $A, B, \ldots, N,$ each party having a measurement device with different possible measurement settings. A source produces particles which are distributed to the individual parties. As customary in this context, it is assumed that each party is in an isolated location, i.e., posterior to each party receiving its particle, no further information can be exchanged among them, at least not without their agreement.

\paragraph*{Quantum framework:}
We assume that quantum theory provides a valid description of the source and measurements used in the Bell test. To each party $k \in \{A, B, \ldots, N\}$ is associated a local Hilbert space $\hil_k$ and the global Hilbert space is $\hil := \hil_A \tot \hil_B \tot \ldots \tot \hil_N$. The particles are described by a quantum state $\Psi$ which is in general a density operator (i.e., a mixed quantum state) $\Psi \in S(\hil)$, where $S(\hil)$ denotes the class of positive semi-definite unit-trace operators on $\hil$. The local spaces are denoted as $S(\hil_{k})$, and we usually call $\rho$ a state that is received and measured by a local detector, $\rho \in S(\hil_{k})$. The measurement devices are described in terms of general positive operator valued measures (POVMs) and we assume that outcomes in a sequence of measurement rounds are independent and identically distributed (i.i.d.). When one of the parties has only partial information about a measurement device, she will possess a more coarse-grained POVM description. We avoid the use of coarse-grained descriptions, as these open the side to attacks, see Section~\ref{sec:crypto}.

\paragraph*{Inputs and outputs:}
Each party has full control over the setting $\x \in \X$ which can be chosen, e.g., by turning a knob or via digital control. For each measurement, there is a set of possible outcomes $\a \in \A$. In the following, we will be mainly interested in the case where both the set of inputs and of outputs are finite\footnote{This is customarily the setting used in Bell tests, although cases with continuous outputs have also been considered~\cite{He10}. Note that the results we present do not depend on the size of the classical input and output spaces; hence, they can be generalized to continuous-variable cases.}. For instance, in the standard Clauser-Horne-Shimony-Holt (CHSH) inequality~\cite{CHSH69,CH74} the set of inputs is $\X = \{0,1\}$ and the possible outcomes are $\A = \{+1,-1\}$, for party A's device and similarly for party B's device. We suppose that along the set of ``good'' outcomes, labelled by an index $\a \in \A$, there is a special label $\nc$ for the ``no-click'' event (i.e., a failed detection in correspondence of a trigger signal) so that the entire alphabet of outcomes becomes $\A' := \{\nc\} \cup \A$.

\paragraph*{Fair sampling for multiple parties:}
In the following presentation we will often restrict our attention to a single measurement device and to a local Hilbert space (so we can take, e.g., $\hil \equiv \hil_{A}$). In fact, fair sampling is an assumption on a single apparatus, i.e., can be considered outside of the context of Bell inequalities. For a Bell test involving $N$ measurement devices, we can say that the setup satisfies fair sampling when each of the $N$ devices, considered individually, satisfies the assumption. However, one obtains a completely equivalent mathematical formulation by considering the $N$ devices as a single collective measurement device split across $N$ locations, and then requiring that this global multipartite device satisfies fair sampling. Hence, all the considerations made for a single device will hold also for the global setup, provided that the fair sampling assumption holds for each device.

\section{Quantum model of a lossy detector}
\label{sec:quantum_model}

We now provide a description of a \emph{measurement device} having non-unit efficiency. We show in particular that such a measurement can be seen as the combination of a filter and an ideal unit-efficiency measurement. We then link this description to the efficiency of the device.

\subsection{POVM elements}
\label{sec:POVM_elements}

We describe the measurement device $\cM$ as a set of POVM elements $\{M_\a := \sum_{\x \in \X} \proj{\x} \tot M^\x_\a\}$ ``taking'' both $\x$ and $\rho$ and returning the outcome $\a$ with a certain probability ${\Pr}_\cM(\a\,|\, \x , \rho) $:
\begin{align}
	\cM: \ S(\;\hil_\X \; \otimes \hil\,) 
	\; & \to \;
	\cD(\A') \nn\\
	\proj{\x} \otimes \, \rho ~~~
	& \mapsto \; 
	\big[\, {\Pr}_\cM(\a\,|\, \x , \rho) = \Tr(M_\a \, \proj{\x} \tot \rho) \,\big]_{\a\in\A'}
	\;.
\end{align}
$\hil_\X$ is a Hilbert space with an orthonormal basis $\{\ket{\x}\}_{\x\in\X}$ and $\cD(\A')$ is the set of probability distributions on $\A'.$ For each $\a \in \A'$, the channel $\cM_\a^\x(\,\cdot\,) := \Tr(M^\x_\a \;\cdot\,)$ is a real-valued completely-positive (CP) map~\cite{NC02} satisfying $\sum_{\a \in \A'} \Tr(M_\a^\x \rho) = 1$ for all density operators $\rho.$
We introduce two flags $\f = \cm$ and $\f = \nc$ associated to two complementary operators $M_\cm^\x$ and $M_\nc^\x$ which satisfy
\begin{align}
	\forall \x \in \X \qquad 
	& M_\cm^\x \; := \; \sum_{\a\in\A} M_\a^\x \;, \\
	\text{so~that} \qquad
	& M_\cm^\x + M_\nc^\x 
	\; = \; \sum_{\a\in\A'} M_\a^\x
	= \; \Id \;.
\end{align}
As we will show later, a description of the operators $M_\cm^\x$ (or equivalently of $M_\nc^\x$) is all that is required to know if a device satisfies fair sampling, while he full POVM description $M_\a^\x$ is not necessary to this end.

\subsection{Filters}
\label{sec:filters}

We can now introduce the notion of \textit{filtering}: for any lossy device $\cM$, one can always find a corresponding pair of completely-positive trace-preserving (CPTP) maps 
\begin{align}
\label{eq:M_decomposition}
	\cF:& ~~~~ S(\,\hil_\X \otimes \hil\,) ~~~~
	\; \to \;
	 S(\,\hil_\f \otimes \hil_\X \otimes \hil \,) \nn\\
	\ocM:& \; S(\,\hil_\f \otimes \hil_\X \otimes \hil\,)
	\; \to \;
	\cD(\A')
\end{align}
such that $\cM = \ocM \circ \cF$. Here, $\hil_\f = \t{span}(\ket{\nc},\ket{\cm})$, $\cF$ represents the filter operation, and $\ocM$ a lossless measurement device. Notice that for mathematical consistency of the equation $\cM = \ocM \circ \cF$ the set of outputs of $\ocM$ must be equal to the set of outputs of $\cM$ and thus also include a ``no-click'' outcome. We therefore consider $\ocM$ to be lossless if a flag $\nc$ has the only effect of ``disabling'' $\ocM$, i.e., $\ocM$ outputs $\nc$ if and only if the filter $\cF$ has already produced a $\nc$ flag.

We remark that the factorization $\cM = \ocM \circ \cF$ is not unique, since we have freedom in redefining $\cF$ and $\ocM$. For instance, we can consider any invertible linear operator $\mathcal{A}$ such that $\ocM' := \ocM \circ \mathcal{A}^{-1}$ and $\cF' := \mathcal{A} \circ \cF$ are both CPTP maps and obtain $\cM = \ocM \circ \cF = \ocM' \circ \cF'$. Observe that complete-positivity is surely preserved if $\mathcal{A}$ is a unitary channel, thus the decomposition $\cM = \ocM \circ \cF$ is certainly not unique.

Given a measurement device $\cM$ with POVM elements $\{M_\a\},$ one can construct a decomposition $\cM = \ocM \circ \cF$ as follows\footnote{We have in this case $M_\a = \sum_{\x \in \X} \proj{\x} \otimes M^\x_\a$, but the construction of the decomposition $\cM = \ocM \circ \cF$ does not rely on this property.}. We have $M_\cm := \sum_{\a\in\A} M_\a$, so that $M_\cm + M_\nc = \Id$. Then, we define the filter $\cF$ as a CPTP map having Kraus operators~\cite{NC02}:
\begin{align}
\begin{split}
	F_\nc \, & := \;
	\ket{\nc} \otimes \sqrt{M_\nc} \,, \\
	F_\cm \, & := \;
	\ket{\cm} \otimes \sqrt{M_\cm} \,, \label{eq:Kraus_ok}
\end{split}
\end{align}
which satisfy the completeness relation $F_\cm^\dag F_\cm + F_\nc^\dag F_{\o} = \Id$. Correspondingly, $\ocM$ is defined as the measurement device acting on the output of $\cF$ (including the flag) and having POVM elements:
\begin{align}
	\label{eq:lossless_M}
	\forall \a \in \A \qquad
	\oM_\a \; & = \;
	\proj{\cm} \otimes (M_\cm)^{-1/2} \, M_\a \, (M_\cm)^{-1/2} 
\end{align}
where the inverse square roots are defined on the support of the operators\footnote{I.e., an Hermitian operator $M = \sum_j p_j \ket{\psi_j}\!\bra{\psi_j}$, with $p_j \geq 0$, is mapped to $(M)^{-1/2} = \sum_j f(p_j) \ket{\psi_j}\!\bra{\psi_j}$, with $f(p) = 1/\sqrt{p}$ if $p>0$ and $f(0) = 0$.}. One can immediately verify that $\sum_{\a\in\A} \oM_\a = \proj{\cm} \otimes \Pi_\cm$, where $\Pi_\cm$ is the projector on the support of $M_\cm$. This means that $\ocM$ has unit efficiency when it receives $\f = \cm$ and a state $\rho$ in the support of $M_\cm$.

From these definition, we can immediately show that $\cM = \ocM \circ \cF$ ensues. In fact, using the shorthand $\xi := \proj{\x} \otimes \rho$, we have:
\begin{align}
	\cF(\xi) \; & = \;
	F_\cm \,\xi\, F_\cm^\dag +
	F_\nc \,\xi\, F_\nc^\dag \nn\\
	& = \;
	\proj{\cm} \otimes \sqrt{M_\cm} \,\xi\, \sqrt{M_\cm} \; + \;
	\proj{\nc} \otimes \sqrt{M_\nc} \,\xi\, \sqrt{M_\nc} \;,
\end{align}
and thus, using the cyclicity of the trace:
\begin{align}
	{\Pr}_{\ocM \circ \cF} (\a\,|\,\xi) 
	\; & = \;
	\Tr\!\left[(M_\cm)^{-1/2} M_\a (M_\cm)^{-1/2} \, \sqrt{M_\cm} \,\xi \,\sqrt{M_\cm} \right] 
	\nn\\
	& = \;
	\Tr\big(\Pi_\cm M_\a \,\Pi_\cm \, \xi \big)
	\; = \;
	\Tr\big(M_\a \,\xi \big) 
	\nn\\[1mm]
	& = \;
	{\Pr}_{\cM} (\a\,|\,\xi) 
\end{align}
for all $\a\in\A$. Consequently, we also have ${\Pr}_{\ocM \circ \cF} (\nc\,|\,\xi) = {\Pr}_{\cM} (\nc\,|\,\xi)$.

\subsection{Efficiency}

The operator $M_\cm^\x$ allows to compute the \textit{efficiency} $\cE$ of the device. The efficiency is, per definition, the probability of obtaining a good outcome ($\f = \cm$) when using a setting $\x$ and an input quantum state $\rho$:
\begin{align}
\label{eq:efficiency_def}
	\cE(\x, \rho) \; := \; {\Pr}_\cM(\checkmark\,|\, \x, \rho) 
	\; = \;
	\Tr(M_\cm^\x \, \rho )
	\; = \;
	\Tr\big[F_\cm \, (\,\proj{\x} \otimes \rho \,) \, F_\cm^\dag\big] \;.
\end{align}
The last equality shows that the efficiency $\cE(\x, \rho)$ can be computed from just the specification of the filter, rather than requiring the full POVM description of $\cM$. Moreover, the efficiency of the detector is a physical property which can be assessed experimentally, provided access to a reliable source of quantum states $\rho \in S(\hil)$. Hence, the efficiency does not depend upon the specific decomposition $\cM = \ocM \circ \cF = \ocM' \circ \cF'$.

The efficiency of the detector allows one to compute the post-selected outcome probabilities
\begin{align}
\label{eq:post-selected-prob}
	{\Pr}_\textup{p.s.} (\a \, | \, \x, \rho)\; := \;
	\frac{{\Pr} (\a \,|\, \x, \rho)}{{\Pr} (\cm | \,\x, \rho)}	\;,
\end{align}
where $\a \in \A$ are ``good'' outcomes, so that $\sum_{\a\in\A} {\Pr}_\textup{p.s.} (\a | \x, \rho) = 1$.

\section{The fair sampling assumption}
\label{sec:fair_sampling}

The fair sampling assumption is a restriction on the physical models of lossy detectors. We introduce a definition which is equivalent to the one introduced by Berry et.\ al.~\cite{Berry10}. We then discuss the consequences and the applications of this definition.

\begin{definition}[Fair sampling]
\label{def:fair_sampling}
We say that a lossy measurement device $\cM$ satisfies the (weak) \emph{fair sampling assumption} if there exists a decomposition $\cM = \ocM \circ \cF$ (as specified in Section~\ref{sec:filters}) whereby the filter $\cF$ factorizes in a part $\cF_C: S(\hil_\X) \to S(\hil_\f \otimes \hil_\X)$ acting on the \textit{classical} setting $\x$ and in a part $\cF_Q:S(\hil) \to S(\hil_\f \otimes \hil)$ acting on the \textit{quantum} input $\rho$. That is, we require:
\begin{align}
	\cF\big( \proj{\x} \otimes \rho \big) \; = \;
	\wedge \left[\;
	\cF_C\big( \proj{\x} \big) \otimes \cF_Q \big(\rho \big) 
	\;\right]\;,
	\label{eq:def_filter}
\end{align}
where the function $\wedge$ (logical AND) acts only on the flags and it means that the filter $\cF$ returns $\cm$ if and only if both $\cF_C$ and $\cF_Q$ return $\cm$.
\end{definition}

A few remarks are now in order. First and foremost, the main reason for using this definition stems from Proposition~\ref{prop:fair_sampling}. There, we show the following strong result: if the factorization of Eq.~\eqref{eq:def_filter} holds, for a given (real) quantum experiment involving lossy detectors and post-selection, there is another (ideal) quantum experiment that involves lossless detectors (and possibly a different quantum state) that exactly reproduces the post-selected statistics of the real experiment. In other words, the post-selected statistics are physical.

Second, in Definition~\ref{def:fair_sampling} we have a decomposition of a filter in two sub-filters. In more general scenarios involving $N$ devices and $2N$ sub-filters, we assume that a successful round of the experiment is obtained only when all filters return $\f = \cm$. In fact, in most quantum experiments a single detection failure is sufficient to invalidate the current observation round\footnote{An exception is provided by experiments on loss-tolerant quantum error correcting codes, which are specifically designed to retain quantum information in a subspace even when some of the physical quantum systems are lost.}.

Third, although it is not possible to verify if a measurement device satisfies fair sampling in a device-independent way, fair sampling does impose restrictions on the outcome probabilities. If these conditions are violated, the experimenter can directly conclude that fair sampling does not hold. See Appendix~\!\ref{app:necessary} for details.

Last, the classical filter is a stochastic map which, given an input $\proj{\x}$ can generate a probabilistic mixture of settings, i.e., the ``real'' setting is $\sum_{\x'} \Pr(\x'|\x)\proj{\x'}$. As a physical example, the experimenter may electronically set a rotation angle $\theta_0$ of a polariser, but the polariser actually rotates by an angle $\theta = \theta_0 + \delta \theta$ (with $\delta \theta$ small and stochastic)~\cite{Rosset12}. However, without loss of generality, we may assume that $\cF_C$ does not change the setting $\x$, i.e.\ it only assigns different success probabilities to different settings. In fact, we can simultaneously re-define $\ocM$ and $\cF_C$ so that the definition of the physical device $\cM$ remains unchanged while the classical filter takes the form $\cF_{C,\cm}(\proj{\x}) = \Pr(\cm |\, \x) \proj{\x}$, as we show in Appendix~\!\ref{app:filter}.

\begin{proposition}[Equivalent formulations of fair sampling]
\label{prop:equivalence}
Consider a lossy measurement device $\cM$ (as formalized in Section~\ref{sec:quantum_model}). The following three properties are equivalent:
\begin{enumerate}
\item $\cM$ satisfies fair sampling, as given in Definition~\ref{def:fair_sampling}.

\item The efficiency of the device factorizes as:
\begin{align}
	\cE( \x, \rho ) \; = \; \cE_C(\x) \, \cE_Q(\rho)
	\label{eq:def_efficiency}
\end{align}
for some real-valued functions $\cE_C$ and $\cE_Q$. This is the definition that was given in Ref.~\cite{Berry10}.

\item The POVM element $M_\cm$, associated to the set of good outcomes of $\cM$, factorizes as in 
\begin{align}
	M_\cm \; =\; M_{C,\cm} \otimes M_{Q,\cm}
	\label{eq:def_POVM}
\end{align} 
where $M_{C,\cm}$ acts on the classical setting and $M_{Q,\cm}$ acts on the quantum input. Moreover,  $M_{C,\cm} = \sum_{\x} \cE_C(\x) \proj{\x}$ is a diagonal matrix, equivalently, $M_\cm^\x = \cE_C(\x) M_{Q,\cm}$ and therefore $M_\cm^\x$ and $M_\cm^{\y}$ are proportional for all $\x, \y \in \X$.
\end{enumerate}
\end{proposition}

\begin{proof}

Eq.~\eqref{eq:def_filter} immediately implies Eq.~\eqref{eq:def_efficiency}, since:
\begin{align}
	\cE(\x,\rho) 
	\; & = \;
	{\Pr}_{\ocM \circ \cF} (\cm | \x,\rho) \; 
	= \;
	{\Pr}_{\cF} (\cm | \x,\rho) \nn\\
	& = \;
	{\Pr}_{\cF_C} (\cm | \x) \, {\Pr}_{\cF_Q} (\cm | \rho) \nn\\
	& \equiv \;
	\cE_C(\x) \, \cE_Q(\rho) \;.
\end{align}

To show that Eq.~\eqref{eq:def_efficiency} implies Eq.~\eqref{eq:def_POVM}, we use the fact that the efficiency $\cE: S(\hil_\X \otimes \hil ) \rightarrow [0,1]$ is a CP linear map, hence $\cE_C$ and $\cE_Q$ are also linear maps (over probabilistic mixtures of inputs) and we can assume, without loss of generality, that both maps take values in $[0,1]$. Since any CP linear map $\cE$ taking value in $[0,1]$ can be written as $\cE(\,\cdot\,) = \Tr(M \,\cdot\,)$, for some matrix satisfying $0 \preccurlyeq M \preccurlyeq \Id$, we have:
\begin{align}
	\cE(\x,\rho) \; & = \; 
	\cE_C(\x) \, \cE_Q(\rho) \nn\\
	& \equiv \;
	\Tr(M_{C,\cm}\proj{\x} ) \, \Tr(M_{Q,\cm}\rho) \nn\\
	& = \;
	\Tr\big[\, (M_{C,\cm} \otimes M_{Q,\cm}) \, (\proj{\x} \otimes \rho ) \, \big] \;,
\end{align}
for some matrices $M_{C,\cm}$ and $M_{Q,\cm}$. Hence, defining $M_\cm := M_{C,\cm} \otimes M_{Q,\cm}$, Eq.~\eqref{eq:def_POVM} holds. Moreover, $M_{C,\cm} = \sum_{\x} \cE_C(\x) \proj{\x}$ is diagonal.

Finally, to show that Eq.~\eqref{eq:def_POVM} implies Eq.~\eqref{eq:def_filter}, we construct filters $\cF_C$ and $\cF_Q$ as in Eq.~\eqref{eq:Kraus_ok}, that is, via the Kraus operators:
\begin{align}
\begin{split}
	F_{C,\cm} \, & := \;
	\ket{\cm} \otimes \sqrt{M_{C,\cm}} \,, \\
	F_{Q,\cm} \, & := \;
	\ket{\nc} \otimes \sqrt{M_{Q,\cm}} \,.
\end{split}
\end{align}
We then set $\cF( \proj{\x} \otimes \rho) :=  \wedge \left[\;\cF_C( \proj{\x} ) \otimes \cF_Q (\rho)\;\right]$, and defining the POVM elements of a lossless device $\ocM$ as in Section~\ref{sec:filters} allows one to show that $\cM = \ocM \circ \cF$ holds. 

\end{proof}

We now provide a result that motivates using the factorization in Eq.~\eqref{eq:def_filter} as definition of fair sampling. The formalism of filters turns out to be very handy: the essence of the proof Proposition~\ref{prop:fair_sampling} is illustrated in Figure~\ref{Fig2}.

\begin{figure}[t]
\begin{center}
	\includegraphics[scale=1.3]{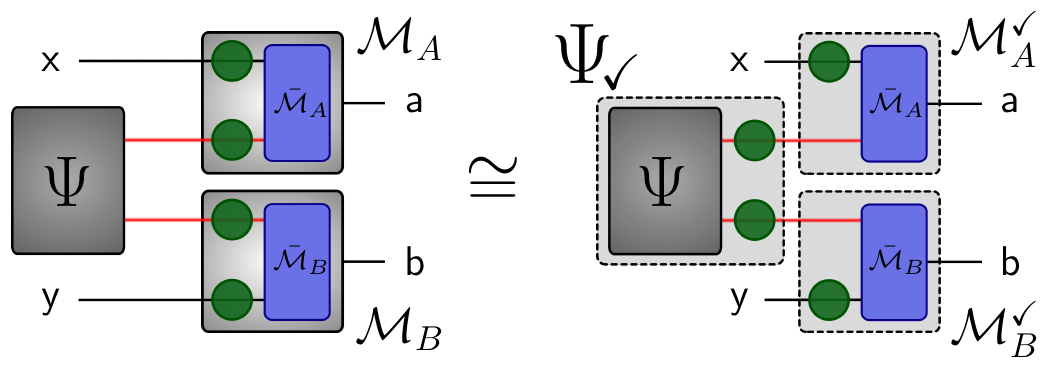}
	\vspace{-3mm}
\end{center}
\caption{Example of a two-party experiment (e.g., a test of a Bell inequality). The picture schematically illustrates how the factorization property of Definition~\ref{def:fair_sampling} leads to fair sampling. }
\label{Fig2}
\end{figure}

\begin{proposition} [Probabilities achievable under fair sampling]
\label{prop:fair_sampling}
Consider any (real) experiment involving a source that produces an $N$-partite quantum state $\Psi \in S(\hil_{1} \otimes \hil_{2} \otimes\ldots \otimes \hil_{N})$, with $N \geq 1$. The state $\Psi$ is measured by a collection of $N$ lossy detectors, $\overrightarrow{\cM} := (\cM_1, \cM_2,\ldots, \cM_N)$, each acting on a local sub-system of $\Psi$ and each satisfying fair sampling as in Definition~\ref{def:fair_sampling}. Call $\vec{\x} := (\x_1,\x_2,\ldots, \x_N)$ the collection of $N$ settings of the devices in a given experimental run, and call $\vec{\a} := (\a_1,\a_2,\ldots, \a_N)$ the collection of $N$ measurement outcomes\footnote{Each party $k$ may have a different set of inputs and outputs, $\x_k \in \X_k$ and $\a_k \in \A_k$.}, each of which could also be a ``no-click'' event $\nc$. A round of the experiment is successful when all the detectors click, and the post-selected statistics is obtained by restricting only to the successful runs.

Then, there is an (ideal) experiment involving a source of a quantum state $\Psi_\cm$ (defined on the same Hilbert space as $\Psi$) and $N$ lossless measurement devices $\oocM := (\cM_1^\cm, \cM_2^\cm,\ldots, \cM_N^\cm)$ whose outcome statistics is equal to the post-selected statistics of the real experiment. That is, we have:
\begin{align}
\label{eq:post-selection_equality}
	{\Pr}^\textup{p.s.}_{\overrightarrow{\cM}}\big(\vec{\a}\,|\,\vec{\x}, \Psi\big)
	\;\equiv\;
	\frac{{\Pr}_{\overrightarrow{\cM}}(\vec{\a}\,|\,\vec{\x}, \Psi)}
	     {{\Pr}_{\overrightarrow{\cM}}(  \cm     |\,\vec{\x}, \Psi)}
	\; = \;
	{\Pr}_{\oocM}\big(\vec{\a}\,|\,\vec{\x}, \Psi_\cm\big) \;,
\end{align}
restricting to the settings $\vec{\x}$ such that ${\Pr}_{\overrightarrow{\cM}}(\cm|\,\vec{\x}, \Psi) \neq 0$\footnote{A setting $\vec{\y}$ having zero acceptance probability effectively can be erased from the set of allowed settings, since $\vec{\y}$ does not appear in the post-selected data. Then, Eq.~\eqref{eq:post-selection_equality} will hold on the restricted collection of settings.}. Moreover, the state $\Psi_\cm$ can be obtained probabilistically from $\Psi$ via probabilistic local operations.
\end{proposition}

\begin{proof}

We decompose each measurement device $\cM_k$, acting on the $k$-th subsystem, as a filter followed by a lossless measurement, $\cM_k = \ocM_k \circ \cF_k$. Moreover, each filter $\cF_k$ factorizes as:
\begin{align}
	\cF_k\big( \proj{\x} \otimes \rho \big) \; = \;
	\wedge \left[\;
	\cF_{k,C}\big( \proj{\x} \big) \otimes \cF_{k,Q} \big(\rho \big) 
	\;\right]\;,
\end{align}
and each sub-filter can herald either a success or a failure:
\begin{align}
\begin{split}
	\cF_{k,C}
	\; & = \; 
	\proj{\cm} \otimes \cF_{k,C,\cm} \, + \;
	\proj{\nc} \otimes \cF_{k,C,\nc} \\
	\cF_{k,Q}
	\; & = \; 
	\proj{\cm} \otimes \cF_{k,Q,\cm} \, + \;
	\proj{\nc} \otimes \cF_{k,Q,\nc} \;.
\end{split}
\end{align}
We then define the normalized quantum state:
\begin{align}
\label{eq:Psi'}
	\Psi_\cm \; := \;
	\frac{1}{\cE_Q(\Psi)} \Big( {\bigotimes}_{k} \, \cF_{k,Q,\cm} \Big) (\Psi ) \;,
\end{align}
where $\cE_Q(\Psi) = \Tr\big[\big( \bigotimes_{k} \cF_{k,Q,\cm} \big) (\Psi)\big]$ is the probability that all the filters $\cF_{k,Q}$ return $\cm$, and we assume $\cE_Q(\Psi) > 0$. The physical interpretation of $\Psi_\cm$ is as follows: apply to the $N$ sub-systems of $\Psi$ the filters $\cF_{k,Q}$ and post-select on all of them returning $\cm$ at the same time. That is to say, it is possible to perform a heralded preparation of $\Psi_\cm$.

Analogously, we can define lossless measurement devices $\cM_k^\cm$ for each party $k$ through the maps $(\cM_k^\cm)^\x_\a:S(\hil_{k}) \to [0,1]$,
\begin{align}
\label{eq:M'}
	\big(\cM_k^\cm\big)^\x_\a \, (\,\cdot\,)
	\; := \;
	\frac{1}{\cE_{k,C} (\x)}
	\Tr\Big\lbrace \big(\bar{M}_k\big)_\a \;
	\big[\, \proj{\cm} \otimes 
	 \cF_{k,C,\cm} (\proj{\x}) \otimes \,\cdot\; 
	\big] 
	\Big\rbrace \;.
\end{align}
Here $(\bar{M}_k)_\a$ are the POVM elements associated to $\ocM_k$ as, e.g., in Eq.~\eqref{eq:lossless_M}, and we assume $\cE_{k,C} (\x) = \Tr[\,\cF_{k,C,\cm} (\proj{\x})\,] > 0$ for all parties\footnote{We can interpret $\cM_k^\cm$ as follows. We repeatedly send some value $\x$ to the filter $\cF_{k,C}$, until this accepts it. In this moment, the measurement device $\cM^\cm_k$ has been ``switched on'' and it is ready to measure an incoming quantum state $\rho$ (without failure): this is a heralded activation of $\cM^\cm_k$.}. We also remark that, as shown in Appendix~\ref{app:filter}, we can assume without loss of generality that the filters do not change the setting, $\cF_{k,C,\cm}(\proj{\x}) = \cE_{k,C}(\x) \proj{\x}$; in this case Eq.~\eqref{eq:M'} simply gives $\cM_k^\cm = \ocM_k$.

As the final step of the proof, we have to compute the probability of obtaining a collection of outcomes $\vec{\a}$ given inputs $\vec{\x}$ when measuring the state $\Psi_\cm$ with the devices $\oocM$, and verify that Eq.~\eqref{eq:post-selection_equality} is satisfied. This is a simple formal manipulation, which is given in Appendix\!~\ref{app:complementary}.

\end{proof}

Proposition~\ref{prop:fair_sampling} is very powerful and general: it ensures that the post-selected statistics obtained using lossy devices is equal to the data one could collect with ideal lossless devices measuring the filtered state $\Psi_\cm$. This has several interesting consequences.

For example, in a scenario where the post-selected data violates a Bell inequality~\cite{CHSH69,CH74}, we can conclude under the fair sampling assumption that the state $\Psi_\cm$ is Bell-correlated \cite{Schmied16}. This implies in particular that $\Psi$ itself contains hidden nonlocality \cite{Gisin96}. Additional information about the state $\Psi_\cm$ can also be obtained from existing device-independent tools: if the Bell violation is high enough, the post-selected data can be used in the framework of self-testing~\cite{MY03, SB19}, to certify directly the exact structure of $\Psi_\cm$ or as a step aiming to certify quantum operations including quantum gates~\cite{SBWS18} or entangling measurements~\cite{RHBS11, BSP18, RKB18}.

We remark, furthermore, that Proposition~\ref{prop:fair_sampling} can be generalized to cases in which some devices are fully characterized, while others are not. In such cases, one needs to appeal to fair sampling only for the not-fully-characterized devices. 
For instance, in \emph{quantum steering} a trusted and an untrusted quantum device jointly measure an entangled quantum state $\Psi$~\cite{CS16}: here, we need to apply fair sampling only to the (single) untrusted device.

Finally, we note that if $\Psi$ is a separable quantum state, then $\Psi_\cm$ is also separable since it can be obtained from $\Psi$ through local probabilistic operations. Consequently, a separable quantum state cannot violate any Bell inequality when using lossy detectors and post-selection, if fair sampling holds. Similarly, under fair sampling it is impossible to violate the quantum bound of a Bell operator with post-selected statistics: in fact, the models and experiments showing that post-selection can lead to incorrect conclusions~\cite{Pearle70, GM87, MP03, PS11, GL11, RG13, JE15, Jogenfors17} do not satisfy fair sampling.

\section{The strong and homogeneous fair sampling assumptions}
\label{sec:strong_fs}

Now we consider two special cases of fair sampling, which we call \emph{strong} and \emph{homogeneous} fair sampling. Incidentally, strong homogeneous fair sampling is what, by and large, is implicitly assumed as the standard definition of fair sampling~\cite{Berry10, WJS98, MMB04}.

\begin{definition}[Strong fair sampling]
\label{def:strong_fair sampling}
Consider a measurement device $\cM$ that satisfies fair sampling, as in Definition~\ref{def:fair_sampling}. The sub-filters $\cF_Q$ and $\cF_C$ herald either a failure or a success, hence they can be written as:
\begin{align}
\begin{split}
	\cF_{C}
	\; & = \; 
	\proj{\cm} \otimes \cF_{C,\cm} \; + \;
	\proj{\nc} \otimes \cF_{C,\nc} \\
	\cF_{Q}
	\; & = \; 
	\proj{\cm} \otimes \cF_{Q,\cm} \; + \;
	\proj{\nc} \otimes \cF_{Q,\nc} \;.
\end{split}
\end{align}

We say that $\cM$ satisfies \emph{strong} fair sampling if there exists a decomposition $\cM = \ocM \circ \cF$ such that the sub-filter $\cF_{Q,\cm}$ is proportional to the identity channel: $\cF_{Q,\cm} = P_{Q,\cm} \textup{Id}$, where $P_{Q,\cm} > 0$ is the success probability. Moreover, if it is possible to write $\cF_{C,\cm} = P_{C,\cm} \textup{Id}$ for some success probability $P_{C,\cm} > 0$, we say that $\cM$ satisfies \emph{homogeneous} fair sampling.
\end{definition}

The notion of \emph{homogeneous} fair sampling is introduced in analogy with the notion of \emph{strong} fair sampling and it means that the detection probability is independent of $\x$. However, not much is to be gained from the homogeneity property. This is because the classical input $\x$ is under full control of the experimenter, therefore it is possible to directly estimate the efficiency $\cE_C(\x)$ from the experimental data (assuming that fair sampling holds) and one can compensate any inhomogeneity in the efficiency. However, in quantum experiments it often happens that devices that (approximately) satisfy fair sampling naturally also possess the homogeneity property. This is the case, e.g., for the \emph{polarization analyser}, a measurement device that we examine in Section~\ref{sec:fs_in_quantum_optics}.

In contrast, \emph{strong} fair sampling does give more stringent guarantees on the experiment. In fact, the state $\Psi_\cm$ which reproduces the post-selected statistics is obtained from the state $\Psi$ after the application of the filters $\cF_{k,Q}$. But if these filters are all proportional to the identity channel, this implies $\Psi_\cm = \Psi$. That is, under the strong fair sampling assumption, the post-selected statistics is a fair representation of the statistics that would be obtained with unit-efficiency detectors.

More in detail, consider again a $N$-local quantum experiment involving $N$ detectors each acting on a part of an entangled quantum state $\Psi$. If all detectors $\overrightarrow{\cM} = (\cM_1, \cM_2,\ldots, \cM_N)$ satisfy strong fair sampling, then we have that the post-selected probabilities are equal to those that one can obtain acting on the actual experimental state $\Psi$ with some lossless devices $\oocM$:
\begin{align}
\label{eq:post-selection_equality_strong}
	{\Pr}^\textup{p.s.}_{\overrightarrow{\cM}}\big(\vec{\a}\,|\,\vec{\x}, \Psi\big)
	\;\equiv\;
	\frac{{\Pr}_{\overrightarrow{\cM}}(\vec{\a}\,|\,\vec{\x}, \Psi)}
	     {{\Pr}_{\overrightarrow{\cM}}(  \cm     |\,\vec{\x}, \Psi)}
	\; = \;
	{\Pr}_{\oocM}\big(\vec{\a}\,|\,\vec{\x}, \Psi\big) \;.
\end{align}
Consequently, if strong fair sampling holds, any claim on $\Psi$ based on the post-selected probabilities will be as good as if the data was produced by lossless devices. That is, the conclusions about non-locality, device-independent entanglement, self-testing, and so on, will hold even if they are based on post-selected data.

However, this strong notion of fair sampling comes at a cost: the conditions specified in Definition~\ref{def:strong_fair sampling} are very stringent and rarely met in actual quantum experiments. We therefore recommend the usage of weak fair sampling as the default definition to be used in experimental settings. The notion of weak fair sampling is more widely applicable and, as argued in Section~\ref{sec:fair_sampling}, it still allows to obtain sharp conclusions based on the post-selected statistics.

Finally, we specialize Proposition~\ref{prop:equivalence} to the case of strong fair sampling.

\begin{proposition}[Equivalent formulations of strong fair sampling]
\label{prop:equivalence2}
Consider a lossy measurement device $\cM$ which satisfies fair sampling, as in Definition~\ref{def:fair_sampling}. Then, $\cM$ satisfies \emph{strong} fair sampling if and only if its efficiency is $\cE(\x,\rho) = \cE_C(\x)$ or, equivalently, if and only if the POVM element associated to $\cm$ has the form $M_\cm = M_{C,\cm} \otimes \Id_Q$.
\end{proposition}

\begin{proof}

It is an immediate consequence of specializing the proof of Proposition~\ref{prop:equivalence} to the cases in which $\cF_Q$ is proportional to the identity map.

\end{proof}

\section{Fair sampling in cryptographic scenarios}
\label{sec:crypto}

Under fair sampling, the post-selected measurement results of a Bell test can also be used for cryptographic tasks, e.g., for random number generation or quantum key distribution~\cite{PABG09, CK11, AM16}. The fair sampling assumption has to be added to the standard \emph{no-leakage} assumption (i.e., no information about the actual measurement is leaked to the adversary) as required also for fully device-independent quantum cryptography \cite{BCK13}. Moreover, device-independent certifications can be used as subroutines in more complex cryptograhic protocols. If a certification is performed employing the fair sampling assumption, then the larger cryptographic protocol will be secure provided that such certification is \emph{securely composable}~\cite{MR09} and, naturally, provided that the fair sampling assumption is valid to begin with.

It is important to note that different parties can have disparate POVM descriptions of the same measurement device, reflecting the degree of knowledge that each party has; in particular, an attacker might have access to some hidden information about the internal functioning of the devices. This also means that parties having different POVM descriptions of one and the same detector might disagree on whether it satisfies fair sampling. We argue that protocols involving lossy devices and post-selection are secure if fair sampling holds unconditionally, and in particular also \emph{according to the POVM description possessed by the adversary}.

\subsection{Explicit attack model}

We consider the so-called \emph{Makarov attack}~\cite{LWW10} as an example where fair sampling appears to hold to the honest parties, but is violated according to the fuller description possessed by the adversary. A detector having settings $\X = \{0,1\}$ and outputs $\A' = \{\nc,+,-\}$ purportedly performs the following lossy measurement:
\begin{align}
\label{eq:traced_out}
\begin{array}{ll}
	M_+^0 \; = \; \frac{1}{4}\proj{0} & \quad
	M_+^1 \; = \; \frac{1}{4}\proj{+} \\
	M_-^0 \; = \; \frac{1}{4}\proj{1} & \quad
	M_-^1 \; = \; \frac{1}{4}\proj{-}
\end{array}	
\end{align}
with $\ket{\pm}:=(\ket{0} \pm \ket{1})/\sqrt{2}$ and thus the no-click events are $M_\nc^0 = M_\nc^1 = \frac{3}{4} \Id$. These measurements are the ones canonically used in CHSH-Bell tests, apart from having efficiency $1/4$, and they seem to satisfy (strong and homogeneous) fair sampling. However, the device is maliciously designed by the adversary: in actuality, it draws uniformly a random variable $r \in \{1,2,3,4\}$, which is unknown to the user, and performs a different measurement depending on $r$. Specifically:
\begin{align}
\begin{array}{r|cccc}
	& r=1 & r=2 & r=3 & r=4 \\
	\hline\\[-9pt]
	M_+^0 &\proj{0}&  0     &  0      &  0       \\[1pt]
	M_-^0 & 0      &\proj{1}&  0      &  0       \\[1pt]
	M_\nc^0  &\proj{1}&\proj{0}&\Id & \Id\\[2pt]
\hline
	M_+^1 & 0      & 0      &\proj{+} &  0       \\[1pt]
	M_-^1 & 0      & 0      &  0      & \proj{-} \\[2pt]
	M_\nc^1  &\Id&\Id&\proj{-} & \proj{+}\\[2pt]
\end{array}
\end{align}
so that tracing out the value of $r$ results in the POVM of Eq.~\eqref{eq:traced_out}. In a cryptographic protocol, the parties must communicate when the detection has succeeded (over a classical authenticated channel). However, this piece of information, together with the knowledge of $r$\footnote{The adversary can know the value of $r$ even under the no-leakage assumption. For instance, a long list of random choices of $r$ could have been stored in a hidden memory inside the device, or they could be obtained algorithmically by a pseudo-random number generator.}, allows the adversary to infer the setting $\x$ and the outcome $\a$ for each successful detection, completely compromising the security. We emphasize that, in this example, from the point of view of the adversary the POVMs describing the detector have the elements 
\begin{align}
\begin{array}{ll}
	M_+^0 \; = \; \proj{0}\otimes\proj{1}_r & \quad
	M_+^1 \; = \; \proj{+} \otimes\proj{3}_r\\
	M_-^0 \; = \; \proj{1} \otimes\proj{2}_r& \quad
	M_-^1 \; = \; \proj{-}\otimes\proj{4}_r,
\end{array}	
\end{align}
which act on the the quantum input of the detector but also on the random variable $r$ which is chosen uniformly at random by the attacker. Manifestly, the POVMs above violate the fair sampling assumption, which corrupts the security guarantees. However, if the variable $r$ is held by the environment (as a physical realization of the loss channel) and can not be accessed by the adversary, her POVM description of the detector is no better than the one of Eq.~\eqref{eq:traced_out}.

Moreover, with two detectors that implement Makarov's attack, it is possible to fake the violation of the CHSH inequality using only separable states, and even saturate the algebraic bound of the operator~\cite{MP03}. That is, calling $\x,\y \in \{0,1\}$ the two inputs and $\a,\b \in \{+,-\}$ the post-selected outputs, these always satisfy $\a\,\b = (-1)^{\x\,\y}$. This can be realized by sending to the two detectors a state $\rho_{r_A,r_B}$ which explicitly depends on the hidden parameters $r_A$ and $r_B$ of the first and second device. A possible choice for $\rho_{r_A,r_B}$ is given in the following table:
\begin{align}
\begin{array}{l|cccc}
	& r_B=1 & r_B=2 & r_B=3 & r_B=4 \\
	\hline\\[-9pt]
	r_A=1 & [0,0]   & \t{vac} & [0,+]   & \t{vac} \\[1pt]
	r_A=2 & \t{vac} & [1,1]   & \t{vac} & [1,-]   \\[1pt]
	r_A=3 & [+,0]   & \t{vac} & \t{vac} & [+,-]   \\[1pt]
	r_A=4 & \t{vac} & [-,1]   & [-,+]   & \t{vac}
\end{array}
\end{align}
where $[p,q]$ is a shorthand for the separable state $\proj{p}\otimes\proj{q}$, while ``vac'' denotes the vacuum state which, naturally, does not trigger the detector; taking into account the detector efficiencies and the vacuum component, the overall detection probability is $1/32 = 1/4 \times 1/4 \times 1/2$. Of course, surpassing the quantum bound of the CHSH operator would reveal that something suspicious is happening. But by adding some white noise in the input states, the expectation value of the CHSH operator can be brought in the region between the classical and quantum limit.

In conclusion, we have to impose that the provided POVM description represents the ``true'' behaviour of the measurement device and that therefore we can be confident that fair sampling holds also according to the adversary's description of the device. In this way, Proposition~\ref{prop:fair_sampling} can be immediately applied to prove security against this adversary, to whom the post-selected statistics are as good as if it were produced by an ideal device. It may seem rather artificial to impose fair sampling in this form, but we argue that there are cryptographic scenarios where this assumption might be plausible. For example, imagine that a trusted company manufactures and sells lossy measurement devices that are guaranteed to work properly, but no information about the internal functioning is provided to the end-users (e.g., as a way to preserve industrial secrets) so that they can only model the device as a black box. A third-party attacker might be able to study a copy of the device, reverse-engineer its functioning, obtain a full POVM description, and discover that fair sampling holds indeed. This would be a setting were Proposition~\ref{prop:fair_sampling} can be put to use and security guarantees recovered, without further assumptions on the structure or calibration of the measurements.

\subsection{Extending to non i.i.d.\ settings}

In this Section we briefly discuss what happens when we do not assume that the measurement statistics are independent and identically distributed (i.i.d.), as done in the rest of the paper. Under no-leakage, the most general attack strategy is obtained when the device internally stores all the quantum states that have been received (and measured) in the past, so that the measurement outcome in a certain round can depend on all the states (and measurement results) at previous rounds. Consequently, the device acts on a Hilbert space whose dimension increases in time and it has a different POVM description at every round. This enables very sophisticated and convoluted attack strategies that are very far from the contexts where fair sampling is a reasonable assumption. Nonetheless, considering fair sampling in non i.i.d.\ settings may be necessary when the measurement devices are not stable in time, e.g.\ as consequence of thermal drift.

Consider now running a Bell experiment where the state preparation and the measurements are repeated for $T$ rounds, in order to collect statistics, and remove the i.i.d.\ assumption, i.e., the devices may behave differently in each of the $T$ rounds (and thus the outcomes are not identically distributed) and there may be memory effect (so that the outcomes are not independent). Are the results of Proposition~\ref{prop:fair_sampling} applicable in these cases?

We remark that the fair sampling assumption can be extended to the non i.i.d.\ case in different ways. To keep the analysis simple, here we propose a redefinition of fair sampling which allows to extend the results of Proposition~\ref{prop:fair_sampling} in a straightforward manner. More general definitions of fair sampling may still be viable, but they would require a more careful analysis.

\begin{definition}[Fair sampling in non i.i.d.\ settings]
\label{def:fair_sampling2}
Consider a lossy measurement device $\cM^{(T)}$ that is used for $T$ rounds in a protocol or experiment. In a general formulation, $\cM^{(T)}$ consists of $T$ devices $\cM_1, \ldots, \cM_T$ where each device $\cM_t$ has a setting $\x_t$, an outcome $\a_t$ and measures a (entangled) state $\rho_t \in S(\hil_t)$; moreover, $\cM_t$ can access an internal memory storing the classical and quantum information it has received or produced during all previous rounds. Each lossy device $\cM_t$ admits a decomposition $\cM_t = \ocM_t \circ \cF_t$, where $\ocM_t$ is a lossless device and $\cF_t$ a filter. Both $\ocM_t$ and $\cF_t$ can act on all the information available at round $t$.

We say that $\cM^{(T)}$ satisfies fair sampling if and only if each filter $\cF_t$ acts only on $\rho_t$ and $\x_t$ (i.e., it acts as the identity channel on the internal memory of the device) and moreover $\cF_t$ factorizes as $\cF_t = \wedge[\cF_{t,C} \otimes \cF_{t,Q}]$, as in Eq.~\eqref{eq:def_filter}.
\end{definition}

Similarly as done in Proposition~\ref{prop:fair_sampling}, we now consider a Bell experiment involving $N$ physically separated non-communicating devices, used for $T$ consecutive rounds, which collectively measure a global quantum state $\Psi$ having $N\times T$ subsystems. A round $t$ of the experiment is successful if and only if all the $N$ devices outputs a $\cm$ flag simultaneously; we then collect the sequence of successes/failures in a list $\vec{f} = \{\f_1, \ldots, \f_T\}$ with $\f_t \in \{\nc,\cm\}$ for each $t$. Notice that the probability that all rounds are successful is exponentially small in $T$ and will essentially never occur when $T$ is large, as is required for having a significant statistical sample. Thus, we suppose that only $T'$ rounds are successful (with $T' \leq T$) and introduce a filtered state $\Psi_{\vec{\f}}$ which is obtained by applying all the $N \times T'$ local filters corresponding to the successful rounds and tracing out the systems involved in unsuccessful rounds.

Following the line of reasoning given in Proposition~\ref{prop:fair_sampling} we can now see that the post-selected measurement statistics of the real experiment is exactly reproduced by the measurement statistics of an ideal experiment where a filtered state $\Psi_{\vec{\f}}$ is measured by lossless devices\footnote{Notice that the filtered state $\Psi_{\vec{\f}}$ now explicitly depends on the sequence of flags $\vec{\f}$ obtained in the experiment. This is not an issue, since in a device-independent protocol the state $\Psi$ is arbitrary to begin with.}. Therefore, in this non i.i.d.\ version of fair sampling we are able to certify in a device-independent way the properties of a quantum state $\Psi_{\vec{\f}}$ that could have been obtained with local probabilistic operations from the state that was present in the lab. Hence, the security proofs that have been developed for certifiable random number generation~\cite{CK11, AM16, PAM10, CZY16, MS16} and device-independent quantum key distribution~\cite{PABG09, LPM13, VV14, ARV19} immediately extend to these devices that satisfy the non i.i.d.\ version of fair sampling.

\section{Fair sampling in quantum optics experiments}
\label{sec:fs_in_quantum_optics}

In this Section we provide a concrete example of a measurement apparatus commonly used in quantum optics, with the intent of showing a possible real-world application of our formalism and also to better understand the conditions that are required to satisfy weak and strong fair sampling, illuminating also how to deal with the case where losses are present prior to the measurement. Specifically, we consider a \emph{polarization analyser}, a device which allows measuring the polarization of incoming photons in arbitrary directions~\cite{Aspect82,SZ99} and the classical setting is the measurement angle $(\x \equiv \theta)$. The layout of the apparatus is schematically depicted in Figure~\ref{Fig4}.

\subsection{Description of a polarization analyser}

We assume that the photons entering the apparatus are described by bosonic operators $b^\dag_H$ and $b^\dag_V$ for the horizontal and vertical polarization, respectively. These operators satisfy canonical commutation relations $[b_H, b_H^\dag] = [b_V, b_V^\dag] = \Id$. A polarizing beam-splitter (PBS) is followed by two non-photon-number-resolving (NPNR) detectors $D_1$ and $D_2$ having POVM elements
\begin{align}
D_1:
\begin{cases}
	E_\nc^\theta = 
	(1-\eta_1)^{b_\theta^\dag b_\theta}~~~~~\,= 
	R_1^{\hat{N}_\theta} \\
	E_\cm^\theta = 
	\Id - (1-\eta_1)^{b_\theta^\dag b_\theta} = 
	\Id - R_1^{\hat{N}_\theta} 
\end{cases}
~
D_2:
\begin{cases}
	E_\nc^{\theta^\bot} = 
	(1-\eta_2)^{b_{\theta^\bot}^\dag b_{\theta^\bot}}~~~~~\,= 
	R_2^{\hat{N}_{\theta^\bot}}\\
	E_\cm^{\theta^\bot} = 
	\Id - (1-\eta_2)^{b_{\theta^\bot}^\dag b_{\theta^\bot}} = 
	\Id - R_2^{\hat{N}_{\theta^\bot}}.
\end{cases}
\end{align}
Here, $\eta_1 = 1 - R_1$ and $\eta_2 = 1 - R_2$ are the efficiencies of detectors $D_1$ and $D_2$. Note also that $\hat{N}_\theta = b_\theta^\dag b_\theta$ and $\hat{N}_{\theta^\bot} = b_{\theta^\bot}^\dag b_{\theta^\bot}$ are the number operators associated to photons having $\theta$ and $\theta^\bot$ polarization, respectively. The angles $\theta, \theta^\bot$ are chosen using a variable polarization rotator (PR) preceding the PBS and are linked to the horizontal/vertical polarization via
\begin{align}
\bigg(\begin{array}{c}
	b_{\theta}^\dag \\
	b_{\theta^\bot}^\dag
\end{array}\bigg)
\; := \;
\bigg(\begin{array}{rr}
	\cos \theta & \sin \theta \\
	- \sin \theta & \cos \theta 
\end{array}\bigg)
\bigg(\begin{array}{c}
	b_H^\dag \\
	b_V^\dag 
\end{array}\bigg) \;.
\end{align}
 If we assume, for the moment being, that the two detectors have equal efficiency $\eta_1 = \eta_2 = \eta = 1-R,$ the four POVM elements of the polarization analyser are:
\begin{align}
\begin{split}
	M_\nc^\theta \, := \;
    E_\nc^{\theta} \tot E_\nc^{\theta^\bot} \; & = \; 
    R^{\hat{N}_\theta} \, R^{\hat{N}_{\theta^\bot}} \ = \ R^{\hat{N}} \\
	M_1^\theta \; := \;
	E_\nc^{\theta} \tot E_\cm^{\theta^\bot} \; & = \; 
    R^{\hat{N}_\theta} \, \big(\Id-R^{\hat{N}_{\theta^\bot}}\big) \\
    	M_2^\theta \; := \;
    	E_\cm^{\theta} \tot E_\nc^{\theta^\bot} \; & = \; 
    \big(\Id-R^{\hat{N}_\theta}\big) \, R^{\hat{N}_{\theta^\bot}}\\
    	M_{1\&2}^\theta \; := \;
    	E_\cm^{\theta} \tot E_\cm^{\theta^\bot} \; & = \; 
    \big(\Id-R^{\hat{N}_\theta}\big) \, \big(\Id-R^{\hat{N}_{\theta^\bot}}\big) \;.
\end{split}
\end{align}
Here, the tensor product represents the bipartition across detectors $D_1$ and $D_2$. In the first line we have used $[\hat{N}_\theta , \hat{N}_{\theta^\bot}] = 0$ and introduced $\hat{N} = \hat{N}_\theta +\hat{N}_{\theta^\bot}.$

\begin{figure}[t]
\begin{center}
	\includegraphics[scale=1.8]{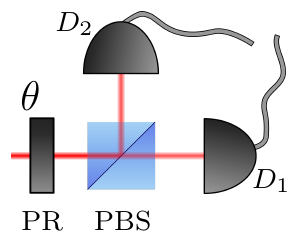}
	\vspace{-3mm}
\end{center}
\caption{A polarization analyser. The polarization of the incoming photons is transformed by a variable polarization rotator (PR); then, the horizontal and vertical components are separated through a polarizing beam-splitter (PBS) in two optical paths; finally, the photons are collected by two non-photon-number-resolving (NPNR) detectors $D_1$ and $D_2$. }
\label{Fig4}
\end{figure}

\subsection{Fair sampling with polarization analysers}

We now show that the polarization analyser previously described (with detectors having equal efficiencies $\eta$ and no dark counts) satisfies the fair sampling assumption. We consider the general case where an arbitrary number of photons can enter the apparatus and, in particular, the input state $\rho$ may have a zero-photon (i.e.\ vacuum) component. Notice that no-click events are caused by two distinct physical processes: either the photons are lost before reaching the apparatus, or some photons reached it, but were not detected. Distinguishing between these two processes is not always possible but also, especially in a device-independent context, not necessary. When dealing with states having a vacuum component, no-clicks would thus occur even while using lossless detectors. Consequently, the ideal experiment reproducing the post-selected data (where no $\nc$ are present) both requires a lossless detector and a filtered state, in which the vacuum component has been discarded.

The relevant Hilbert space is the Fock space consisting of photons that are only distinguishable in their polarization degree of freedom:
\begin{align}
	\hil  \; = \;
	\t{span} \big\lbrace\,
	\big(b_\theta^\dag\big)^{n_{\theta}} \big(b_{\theta^\bot}^\dag\big)^{n_{\theta^\bot}} \ket{\vac}
	\; \big| \;
	n_\theta, n_{\theta^\bot} \in \mathds{N} \,\big\rbrace \;,
\end{align}
where $n_\theta$ and $n_{\theta^\bot}$ are the number of photons polarized in direction $\theta$ and in the orthogonal direction~$\theta^\bot$, respectively. In this Hilbert space, the POVM elements of the polarization analyser are given by\footnote{Here, we have the possibility that both NPNR detectors click at the same time, corresponding to the POVM element $M_{1\&2}^\theta$. If it is required that the device has only two possible measurement outcomes, we can arbitrarily re-assign this outcome to either $M_1^\theta$ or $M_2^\theta$.}
\begin{align}
	M_\nc^\theta \; & = \; 
	R^{\hat{N}} \; = \;
	\proj{\vac} +
	\sum_{n=1}^\infty R^n \; \Pi_n 
	\\
	M_\cm^\theta
	\; & = \;
	M_1^\theta + M_2^\theta + M_{1\&2}^\theta
	\; = \;
	\Id - R^{\hat{N}}
\label{eq:multi-photonM}
\end{align}
where $\Pi_n$ is the ($\theta$-invariant) projector onto the $n$-photon sector of the Fock space. The expression on the right hand side of Eq.~\eqref{eq:multi-photonM} is invariant under change of rotation angle $\theta$: hence, we are in a (homogeneous) fair sampling situation. 

We also remark that $M_\cm^\theta$ has no overlap with the vacuum, $\bra{\vac} M_\cm^\theta \ket{\vac} = 0$, which immediately implies that the filtered state
\begin{align}
	\rho_{\cm} 
	\; = \;
	\frac{1}{\cE_Q(\rho)}
	\sqrt{M_\cm^\theta} \, \rho \, \sqrt{M_\cm^\theta}
\end{align}
also does not have overlap with the vacuum. As a special case, notice that when $\eta=1$ (the detector are lossless) we have $M_\cm^\theta = \Id - \proj{\vac}$ and therefore post-selecting is equivalent to using a filtered state where the vacuum component has been projected out.

\subsection{Strong fair sampling with polarization analysers}

We now look under which conditions the \emph{strong} version of fair sampling holds. From Eq.~\eqref{eq:multi-photonM} it follows that strong fair sampling is formally verified whenever we restrict the Hilbert space to the $k$-photon sector (for any fixed integer $k>0$). However, we consider only the case $k=1$, which already corresponds to a very stringent experimental requirement. That is, we imagine that in each round of the experiment exactly one photon enters the detector. Thus we have: 
\begin{align}
	\hil 
	\; = \; 
	\t{span} \big\{\, 
	\ket{\theta}  = b_\theta^\dag \ket{\vac}, \; 
	\sket{\theta^\bot} = b_{\theta^\bot}^\dag \ket{\vac}\,
	\big\} \;,	
\end{align}
with $\theta$ and $\theta^\bot$ any two orthogonal polarization directions. This yields:
\begin{align}
	M_\nc^\theta 
	\; & = \; 
	R^{\ket{\theta}\bra{\theta}} R^{\sket{\theta^\bot}\sbra{\theta^\bot}} 
	\; = \;
	R^\Id 
	\; = \;
	(1-\eta) \, \Id \nn\\
	M_\cm^\theta 
	\; & = \; 
	\Id - M_\nc^\theta 
	\; = \;
	\eta \, \Id \;. 
\label{eq:1photonM}
\end{align}
The expression on the right hand side of Eq.~\eqref{eq:1photonM} is manifestly independent of $\theta$ and is proportional to the identity: the (homogeneous) strong fair sampling assumption holds in this case.

\subsection{Experimental applicability}

As discussed, the polarization analyser satisfies strong fair sampling only when the input quantum state consists of exactly one photon. This is a very demanding condition that is not met in any experiment that we know of, as currently there is no technology capable of producing individual photons with high fidelity and (almost) unit efficiency \cite{FSS18,SP19}. To this end, the best we can achieve is through the \emph{heralded} preparation of a photon, e.g., by first producing a photon pair and exploit the fact that one of the two photons has triggered a detector to infer the presence of the second one in a separated optical path. Even in this case, though, the applicability of strong fair sampling is debatable, since losses are always present in practice. To this end, we show in Section~\ref{sec:state-dependent} how to extend the results to input states that can only be approximated as a single-photon states.

When the input quantum is not restricted to the single-photon sector, we have shown that the polarization analyser satisfies (weak) fair sampling. This happens, e.g., when the state has a (significant) overlap with the vacuum. However, since the polarization analyser can click only when one or more photons are present in the apparatus, the filtered state does not have overlap with the vacuum.

We could then consider sources that (up to an excellent degree of accuracy) emit at most one photon at a time. In these cases, the \emph{filtered} state will have exactly one photon entering the polarization analyser. We could also consider sources that could emit multiple photons at a time, as is the case in spontaneous up- and down-conversion processes in non-linear crystals. In these cases, the filtered quantum state does not contain the vacuum, but it does contain multi-photon components. Noticing that the more photons enter the device, the higher is the probability of a detection we deduce that, correspondingly, the filtered state $\rho_\cm$ has more weight on large photon-number components, compared to the experimental state $\rho$.

Finally, we can give a concrete recommendation for proof-of-principle demonstrations of device-independent protocols relying on post-selected data. The measurement device should be extensively tested and calibrated. Subsequently, one should publish the best available estimation of the efficiency of the device (the POVM elements $M_\cm^\theta$), together with the uncertainty to which they are known. From an abstract perspective, this is less information than having the individual POVM description of $M_1^\theta, M_1^\theta$, and $M_{1\&2}^\theta$, which would be required in device-specific protocols (such as, say, the BB84 protocol). Moreover, in some cases the estimation of $M_\cm^\x$ might be experimentally  easier than the estimation of all $M_\a^\x$ elements, for instance when the number of possible outcomes (the size of the set $\A$) is large. Provided that the fair sampling condition is satisfied, this approach provides the means to demonstrate the fundamental elements of device-independent protocols with post-selected data, avoiding the stringent requirement of having a extremely high detector efficiency.

\section{Approximate fair sampling}
\label{sec:approximate_fair sampling}

Here we investigate what happens when the fair sampling assumption is not satisfied exactly. Intuitively, we expect that the probability distributions that can be obtained after post-selection are close to those that can be obtained measuring the filtered state $\rho_\cm$ with a lossless device. And indeed, a small deviation from the exact fair sampling condition induces a small perturbation in the post-selected probability distribution (with linear dependence, to the leading order). In turn, the fact that the post-selected probabilities are close to the ideal ones implies, e.g., that the experimental expectation values of Bell operators are also close to the ideal value.

\subsection{Bound on the total variation distance}

As a starting point, we recall that $M_{C,\cm}$ is diagonal, hence we have $M_\cm = \sum_{\x} \cE_C(\x) \proj{\x} \otimes M_{Q,\cm}$ or, equivalently, $M_\cm^\x = \cE_C(\x)  M_{Q,\cm}$. We can rewrite this as:
\begin{align}
\label{eq:starting_point}
	\cE_C(\x) \,\Pi_\cm
	\; = \;
	M_{Q,\cm}^{-1/2} M_\cm^\x M_{Q,\cm}^{-1/2}  
\end{align}
where $\Pi_\cm$ is the projector onto the support of $M_{Q,\cm}$ (and on the support of $M_{Q,\cm}^{-1/2} M_\cm^\x M_{Q,\cm}^{-1/2}$ as well) and thus we have
\begin{align}
	 \cE_C(\x)
	 \; = \;
	 \norm{M_{Q,\cm}^{-1/2} M_\cm^\x M_{Q,\cm}^{-1/2}} \;.
\end{align}
Here and in the following, the norms are operator norms. We then consider small deviations from fair sampling as expressed in Eq.~\eqref{eq:starting_point}.

\begin{proposition}[Approximate fair sampling]
\label{prop:approximate_fs}
Consider a lossy device $\cM$ having POVM elements $M_\a^\x$ and $M_\cm^\x = \sum_{\a\in \A} M^\x_\a$. We suppose that there exists a positive semi-definite operator $M_{Q,\cm}\succcurlyeq 0$ such that each $M_\cm^\x$ satisfies, for some $\epsilon \in [0,1)$:
\begin{align}
	\norm{\,
	\Pi_\cm - 
	\frac {M_{Q,\cm}^{-1/2} M_\cm^\x M_{Q,\cm}^{-1/2}}
	{\norm{M_{Q,\cm}^{-1/2} M_\cm^\x M_{Q,\cm}^{-1/2}}}
	\,}
	\; \leq \;
	\epsilon
	\label{eq:approx}	
\end{align}
where $\Pi_\cm$ is the projector on the support of $M_{Q,\cm}$. We also assume $\Pi_\cm M_\cm^\x \Pi_\cm = M_\cm^\x$ for all $\x \in \X$ and $\Pi_\cm \rho \, \Pi_\cm \neq 0$ for the experimentally relevant input states.

Then, there exists a normalized quantum state $\rho_\cm$, which can be probabilistically filtered from $\rho$, and there exists a lossless device $\ocM$ such that:
\begin{align}
	\sum_{\a\in\A} \frac{1}{2}
	\left\vert\, 
	\frac{\Tr(M_\a^\x \rho)}{\Tr(M_\cm^\x \rho)} -
	\Tr(\bar{M}_\a^\x \rho_\cm)
	\,\right\vert
	\; \leq \;
	\frac{\epsilon}{1 - \epsilon}
	\label{eq:approx_result}	
\end{align}

When the measurement device $\cM$ satisfies inequality~\eqref{eq:approx} we say that the $\epsilon$-\emph{approximate} (weak) fair sampling assumption holds.
\end{proposition}

\begin{proof}
See Appendix\!~\ref{app:proof}.
\end{proof}

A few remarks are in order. First, in the limit $\epsilon \to 0$ approximate fair sampling becomes the condition for (exact) fair sampling. Second, the left hand side of inequality~\eqref{eq:approx_result} is the total variation distance between the post-selected probability distribution ${\Pr}_{\cM}^{\t{p.s.}}(\a|\x,\rho)$ and the distribution ${\Pr}_{\ocM}(\a|\x,\rho_\cm)$ obtained by measuring the filtered state $\rho_\cm$ with a lossless detector\footnote{The total variation distance is the classical analogue of the trace distance for density matrices. It expresses the expected distinguishability between two distributions using random samples from one of the distributions.}. Finally, we can specialize the Proposition to approximate \emph{strong} fair sampling by consider the special condition $M_{Q,\cm} = \Id$. In this case, Eq.~\eqref{eq:approx} simply becomes:
\begin{align}
	\norm{\,
	\Id - 
	\frac {M_\cm^\x}
	{\norm{M_\cm^\x }}
	\,}
	\; \leq \;
	\epsilon
\end{align}
and the filtered state is equal to the experimental state, $\rho_\cm = \rho$.

\subsection{Approximate fair sampling in multipartite settings}
\label{sec:appr_str_fs_consequences}

Proposition~\ref{prop:approximate_fs} links an inequality regarding the POVM of a measurement device to an inequality on the probability distributions that can be observed using such a device. In a multipartite quantum experiment this probability distribution corresponds to the marginal distribution that is observed by a single detector, ignoring the measurement outcomes of all other devices involved in the experiment. Therefore, Proposition~\ref{prop:approximate_fs} does not (directly) bound the deviation for the global distribution over all possible outcome configurations of the devices. However, this shortcoming can be easily amended, using the following simple idea.

Consider a quantum experiment where $N$ parties are involved, measuring a multipartite state $\Psi$ with lossy devices $\cM_1, \ldots, \cM_N$ which approximately satisfy fair sampling. Then, we can equivalently consider the multipartite measurement device $\overrightarrow{\cM}$ consisting of the collection of these $N$ devices, whose settings (and outputs) consists in the collection of settings (and outputs) of the individual devices. It is straightforward to verify that this joint device then satisfies $\epsilon_\t{tot}$-approximate fair sampling; the total approximation error $\epsilon_\t{tot}$ is essentially given by the sum of the approximation errors for the individual devices. Therefore, for this multipartite device Proposition~\ref{prop:approximate_fs} directly applies. That is, it is possible to bound to the deviation of the post-selected probability distribution from the result of measuring a filtered state $\Psi_\cm$ with $N$ lossless devices. More details can be found in Appendix~\ref{app:multipartite}. As a further consequence, suppose we want to measure the expectation value of a $N$-partite Bell operator $B$. Using the (slightly skewed) post-selected distribution the relative error on $\langle B \rangle$ is upper bounded by $2 \epsilon_\t{tot}$ -- for the proof, see Appendix~\ref{app:X}.

Finally, we remark that if the $N$ separated devices satisfy approximate \emph{strong} fair sampling, then the entire multipartite setting also satisfies approximate strong fair sampling; this implies that the outcome statistics for the collection of devices is close (in total variation distance) to what could be obtained measuring the actual experimental state $\Psi$ with lossless devices. In this case, the global approximation error $\epsilon_\t{tot}$ for the multipartite device is also roughly equal the sum of the errors for individual devices.

\section{Approximate fair sampling in quantum optics}
\label{sec:appr_fs_in_quantum_optics}

In Section~\ref{sec:fs_in_quantum_optics} we have shown that a polarization analyser satisfies fair sampling provided that the two NPNR detectors have equal efficiency $\eta$. Now, we will consider the case where the efficiencies of the NPNR detectors are very similar, but not exactly equal. In this case, we show that the polarization analyser satisfies \emph{approximate} fair sampling; moreover, if the input state contains exactly one photon, \textit{approximate} strong fair sampling holds.

\subsection{Polarization analyser with unequal detectors}
\label{sec:unequal_detectors}

We generalize the quantum modelling of the polarization analyser given in Section~\ref{sec:fs_in_quantum_optics} to the case in which the two detectors have similar (but not exactly equal) efficiencies, $\eta_1$ and $\eta_2$. We assume without loss of generality $\eta_2 \geq \eta_1$ and we write $R = R_2 \leq R_1 = R \,(1+\delta)$, with $\delta \geq 0$. This results in the following four POVM elements:
\begin{align}
\begin{split}
	M_\nc^\theta \, := \;
    E_\nc^{\theta} \tot E_\nc^{\theta^\bot} \; & = \; 
    R^{\hat{N}_\theta} (1+\delta)^{\hat{N}_{\theta}} \, R^{\hat{N}_{\theta^\bot}}
    \; = \; R^{\hat{N}} (1+\delta)^{\hat{N}_{\theta}} \\
	M_1^\theta \; := \;
	E_\nc^{\theta} \tot E_\cm^{\theta^\bot} \; & = \; 
    R^{\hat{N}_\theta} (1+\delta)^{\hat{N}_{\theta}} \,\big(\Id-R^{\hat{N}_{\theta^\bot}} \big) \\
    	M_2^\theta \; := \;
    	E_\cm^{\theta} \tot E_\nc^{\theta^\bot} \; & = \; 
    \big[\Id-R^{\hat{N}_\theta} (1+\delta)^{\hat{N}_{\theta}}\big] \, R^{\hat{N}_{\theta^\bot}} \\
    	M_{1\&2}^\theta \; := \;
    	E_\cm^{\theta} \tot E_\cm^{\theta^\bot} \; & = \; 
    \big[\Id-R^{\hat{N}_\theta} (1+\delta)^{\hat{N}_{\theta}}\big] \, \big(\Id-R^{\hat{N}_{\theta^\bot}} \big) \;.
\end{split}
\end{align}

\subsection{Approximate fair sampling with a polarization analyser}

We now consider photonic states which have a (possibly large) component of vacuum and a certain probability of emitting more than one photon, as common in experiments that use up- or down-conversion processes. The Hilbert space is 
\begin{align}
	\hil  \; = \;
	\t{span} \big\lbrace\,
	\big(b_\theta^\dag\big)^{n_{\theta}} \big(b_{\theta^\bot}^\dag\big)^{n_{\theta^\bot}} \ket{\vac}
	\; \big| \;
	n_\theta, n_{\theta^\bot} \in \mathds{N} \,\big\rbrace \;.
\end{align}
The operator $M_\cm^\theta$ has an explicit (small) dependence from the setting $\theta$: 
\begin{align}
\begin{split}
	M_\nc^\theta 
	\; & = \; 
	R^{\hat{N}} \, (1+\delta)^{\hat{N}_{\theta}} 
	\\
	M_\cm^\theta 
	\; & = \;
	\Id - R^{\hat{N}} \, (1+\delta)^{\hat{N}_{\theta}} \;.
\label{eq:multi-photon_approx}
\end{split}
\end{align}

We now have to find a suitable positive semi-definite operator $M_{Q,\cm}$ and a projector $\Pi_\cm$ in order to show that weak fair sampling is approximately satisfied. We choose:
\begin{align}
\label{eq:M_Q}
	M_{Q,\cm} & := \; \Id - R^{\hat{N}}  \\
	\Pi_\cm   \; & := \; \Id - \proj{\vac} \;.
\end{align}
In Appendix\!~\ref{app:F} we show that, with these definition, the condition
\begin{align}
	\norm{\,
	\Pi_\cm - 
	\frac {M_{Q,\cm}^{-1/2} M_\cm^\theta M_{Q,\cm}^{-1/2}}
	{\norm{M_{Q,\cm}^{-1/2} M_\cm^\theta M_{Q,\cm}^{-1/2}}}
	\,}
	\; & = \; 
	\frac{1 - \eta }{\eta}\, \delta 
\end{align}
holds. Thus, we can directly apply the results on approximate weak fair sampling derived in Proposition~\ref{prop:approximate_fs}. The maximum deviation in the post-selected probability and those obtainable from $\rho_\cm$ is of order $\delta \, (1-\eta)/\eta$, in total variation distance.

We can also characterize the filtered state $\rho_\cm$ that allows to (approximately) reproduce the post-selected probabilities:
\begin{align}
	\rho_\cm 
	\; = \;
	\frac{1}{\cE_{Q}(\rho)}
	\sqrt{M_{Q,\cm}} \, \rho \, \sqrt{M_{Q,\cm}}
	\; = \;
	\frac{1}{\cE_{Q}(\rho)}
	\sqrt{\Id - R^{\hat{N}}} \, \rho \, \sqrt{\Id - R^{\hat{N}}} \;.
\end{align}
In particular, the state $\rho_\cm$ has no overlap with the vacuum.

\subsection{Approximate strong fair sampling with a polarization analyser}

In this Section, we restrict to cases where exactly one photon is entering the apparatus, i.e., we assume that we are employing an (almost) perfect single-photon source. The Hilbert space is
\begin{align}
	\hil 
	\; = \; 
	\t{span} \big\{\, 
	\ket{\theta}  = b_\theta^\dag \ket{\vac}, \; 
	\sket{\theta^\bot} = b_{\theta^\bot}^\dag \ket{\vac}\,
	\big\} \;,	
\end{align}
and the POVM elements associated to a failed/successful detection are:
\begin{align}
\begin{split}
	M_\nc^\theta \; & = \; 
	R \, (1+\delta)^{\ket{\theta}\bra{\theta}} \; = \;
	(1-\eta) \, \big(\Id + \delta \proj{\theta}\big) \\
	M_\cm^\theta \; & = \;
	\Id - M_\nc^\theta 
	~~~~~~~\: = \;
	\eta \, \Id  - (1-\eta) \delta \proj{\theta} \;. \label{eq:1photonA}
\end{split}
\end{align}
The POVM element $M_\cm^\theta$ has an explicit (small) dependence on $\theta$. In fact, we can explicitly show that the condition for approximate strong fair sampling holds. Assuming $\delta \geq 0$, we obtain
\begin{align}
\forall \theta \qquad
	\big|\big|{M_\cm^\theta}\big|\big|
	\; = \; 
	\norm{\,\eta \, \Id  - (1-\eta) \delta \sproj{\theta^\bot}\,}
	\; = \; \eta  \;,
\end{align}
and then
\begin{align}
	\norm{ \Id - \frac{M_\cm^\theta}{\norm{M_\cm^\theta}}} \; = \;
	\norm{ \Id - \frac{M_\cm^\theta}{\eta}} \; = \;
	\norm{ \frac{1-\eta}{\eta}\, \delta \, \sproj{\theta^\bot}} \; = \;
	\frac{1-\eta}{\eta}\, \delta \;.	
\end{align}

This result directly allows us to employ Proposition~\ref{prop:approximate_fs}, specialized to the case of strong fair sampling. We can conclude that the post-selected statistics obtained with this device can have a deviation (in total variation distance) at most of order $\delta \, (1-\eta)/\eta$ from the statistics obtained measuring the single-photon state with a lossless device. 

As previously discussed, single photon sources are not achievable with current technology, heralded photon preparation being the method that comes the closest to this goal. Therefore, in the next Section we analyse what happens in the case in which the input state is only approximately a single-photon state and has a small vacuum and multi-photon component.

\section{Assumptions on the input state}
\label{sec:state-dependent}

In this final Section, we consider the case in which the quantum state has a small component not belonging to the reference Hilbert space. We show that the outcome statistics is not overly affected by the presence of such perturbations. Next, we extend the fair sampling assumption to the case where the factorization of Eq.~\eqref{eq:def_filter} holds only for the specific state $\Psi$ that is employed in the (multipartite) quantum experiment. We call this the \emph{state-dependent} fair sampling assumption.

\subsection{Effects of imperfections in state preparation}

Consider an imperfectly prepared quantum state having the form
\begin{align}
	\hat\rho
	\ & \in \
	S(~~~\hil~~~ \oplus ~~~\hil^\bot~~) 
	\nonumber\\
	\hat\rho 
	\; & = \, 
	\left(
	\begin{array}{cc}
	(1-\epsilon') \, \rho & c \\
	c^\dag & \epsilon' \, \rho^\bot	
	\end{array}
	\right)
	\;,
\label{eq:rho_hat}
\end{align}
where $\rho$ and $\rho^\bot$ are normalized quantum states and $c$, $c^\dag$ are the coherence terms. The Hilbert space $\hil$ represents the set of states for which the measurement device has been calibrated, and $\hil^\bot$ is an orthogonal space containing all other degrees of freedom. The fair sampling assumption holds for states $\rho \in S(\hil)$, while the behaviour of the measurement device for an input state $\rho^\bot \in S(\hil^\bot)$ can be arbitrary.

As a concrete example, consider a photonic state $\hat\rho$ entering a polarization analyser. Strong fair sampling holds when the input consists of exactly one photon, but the experimental state $\hat\rho$ also possess a small component $\epsilon'\rho^\bot$ which overlaps with the vacuum and with multi-photon states. As a second example, a polarization analyser satisfies (weak) fair sampling when $\hat\rho$ consists of photons that are only distinguishable in polarization. That is, the photons must have the correct frequency, waveform, spatial mode, and so on, consistent with the specifications and calibration of the measurement device. However, the experimental state $\hat\rho$ may have a small component $\epsilon' \rho^\bot$ of photons that have different properties, for which the behaviour of the device is not known. If $\epsilon'$ is small, the deviation in the post-selected statistics will also be small; quantitatively, we have the following result.

\begin{proposition}
\label{prop:state}
Consider a quantum state $\hat\rho \in S(\hil \oplus \hil^\bot)$ as in Eq.~\eqref{eq:rho_hat}. Consider a lossy detector having POVM elements $\{\hat{M}_\a^\x\}$ acting on $\hil \oplus \hil^\bot$ and define $M_\a^\x := \Pi \hat{M}_\a^\x \Pi$, $\hat{M}_\cm^\x := \sum_{\a \in \A} \hat{M}_\a^\x$ and $M_\cm^\x := \sum_{\a \in \A} M_\a^\x$. Then, the following trace distance bound holds:
\begin{align}
\label{eq:result2}
	\sum_{\a\in\A} \frac{1}{2}
	\left\vert\, 
	\frac{\Tr(\hat{M}_\a^\x \hat\rho)}{\Tr(\hat{M}_\cm^\x \hat\rho)} -
	\frac{\Tr(M_\a^\x \rho)}{\Tr(M_\cm^\x \rho)}
	\,\right\vert
	\; \leq \;
	\frac{2 \left(\norm{c}_\Tr + \epsilon'\right)}
	{\max 
	\left\lbrace
	\Tr(\hat{M}_\cm^\x \hat\rho)
	\, , \,
	\Tr(M_\cm^\x \rho)
	\right\rbrace}
	\;,
\end{align}
where $\norm{c}_\Tr := \Tr(\sqrt{c^\dag c})$ is the trace norm of $c$, which satisfies $\norm{c}_\Tr \leq \sqrt{\epsilon'(1-\epsilon')}$.
\end{proposition}

\begin{proof}
See Appendix~\!\ref{app:proof2}.
\end{proof}

In general, the deviation in the post-selected statistics can be of order $\sqrt{\epsilon'}$. However, if the state has no coherence terms ($c=0$), then the deviation is only of order $\epsilon'$. Thermal light is an instance of this case, being completely incoherent in the photon-number basis.

We also remark that in the realistic case where both device calibration and state preparation are not perfect, the bounds given in Proposition~\ref{prop:state} and in Proposition~\ref{prop:approximate_fs} can be combined. The deviation from the ideal case is at most the sum of the bounds given in Eq.~\eqref{eq:approx_result} and in Eq.~\eqref{eq:result2}, as can be shown applying the triangle inequality.

\subsection{State-dependent fair sampling}

We suppose that the measurement device $\cM_1 =\ocM_1 \circ \cF_1$ used by the first party of a Bell-like experiment does not satisfy fair sampling (not even approximately). We require, instead, that the application of the successful branch $\cF_{1,\cm}: S(\hil_\X \otimes \hil_{1}) \to S(\hil_\X \otimes \hil_{1})$ to the first component of the state $\Psi \in S(\hil_{1}\tot \cdots \tot \hil_{N})$ results in:
\begin{align}
\forall \x \in \X \qquad
	\big(\cF_{1,\cm} \otimes \t{Id}_\t{R} \big) \,
	\big(\proj{\x} \otimes \Psi \big)
	\; = \;
	\widetilde{\cF}_{1,C,\cm}\big(\proj{\x}\big) \otimes 
	\cE_{1,Q}(\Psi) \, \Psi_{1,\cm}	
\end{align}
where $\t{Id}_\t{R}$ is the identity channel on the rest of the parties, $\Psi_\cm$ is a (normalized) quantum state that can be locally filtered from $\Psi$, and $\cE_{1,Q}(\Psi)$ is proportional to the probability of successfully filtering the state $\Psi_{1,\cm}$ from $\Psi$. That is, there exists a CP map $\widetilde{\cF}_{1,Q,\cm}$ such that:
\begin{align}
	\Psi_{1,\cm} 
	\; = \;
	\frac{1}{\cE_{1,Q}(\Psi)} \,
	\big( \widetilde{\cF}_{1,Q,\cm} \otimes \t{Id}_\t{R} \big) \,
	\big( \Psi \big)	 \;.
\end{align}
Then, the same results as in Proposition~\ref{prop:fair_sampling} apply in this case. Consider, in fact, the post-selection corresponding to first device clicking (we ignore the post-selection applied by the other parties for now). Then, the post-selected probabilities that one observes in the realized quantum experiment can be reproduced by an ideal experiment where the first device is lossless, defined as in Eq.~\eqref{eq:M'}, and the state that has been distributed to the $N$ parties is the filtered state $\Psi_\cm$.

We can then apply the same reasoning to each of the parties involved in the experiment. Suppose that all the measurement devices satisfy \emph{state-dependent} fair sampling; then, the global post-selection (where the devices click contemporaneously) produces probabilities that are equal to those one would obtain when all parties use lossless devices, defined as in Eq.~\eqref{eq:M'}, acting on the normalized quantum state:
\begin{align}
	\Psi_\cm 
	\; = \;
	\frac{1}{\cE_{1,Q}(\Psi) \cdots \cE_{N,Q}(\Psi)} \,
	\big(\widetilde{\cF}_{1,Q,\cm} \otimes \cdots \otimes \widetilde{\cF}_{N,Q,\cm} \big) \,
	\big( \Psi \big)	 \;,	
\end{align}
where each $\widetilde{\cF}_{k,Q,\cm}$ is a local extraction map that depends on $\cM_k$ and on $\Psi$.

We can also consider approximate versions of state-dependent version fair sampling, taking into account detector and state preparation imperfections. The results previously derived for approximate fair sampling would also hold in this case.

\section*{Discussion}

In this paper, we have analysed the effect of post-selection on the outcome statistics of a measurement apparatus in the context of Bell-like experiments. To this end, we have employed a general and fully quantum model of the measurement device, described in terms of POVM elements. We have identified a condition, equivalent to the one of Ref.~\cite{Berry10}, which allows the measurement statistics to exactly correspond to the statistics of some ideal quantum experiment, involving \emph{ideal} lossless detectors and a \emph{filtered} quantum state. The filtered state is (possibly) different from the actual experimental state but can be obtained from it through local probabilistic filters. In simple terms, this condition amount to the fact that the device acts independently on the classical and quantum inputs it receives, conditioned on having had a successful detection. Differently from what is done in most of the previous literature, it is this condition that we identify as the (weak) fair sampling assumption.

We argue that the weak notion of fair sampling is the most useful and applicable one, especially in the context of quantum optics. In fact, in quantum optics experiments the produced photonic states have a component of vacuum, but the common praxis is to characterize properties (such as fidelity, purity, or expectation values of observables) of a quantum state where the vacuum component has been projected out. This exactly corresponds to the filtering operation induced by fair sampling.

To make the result of this work applicable in real experimental settings, we have investigated the effect of small deviations from exact fair sampling. We have shown that imperfections in the calibration of the device result in a deviation in the post-selected probabilities that scales essentially linearly in the size of the miscalibration, but is amplified by a factor inversely proportional to the efficiency of the device. For a polarization analyser to satisfy fair sampling the two photon detectors are required to have the same efficiency, which could be enforced by artificially decreasing the efficiency of one of the detectors. Our analysis of approximate fair sampling takes care of any remaining efficiency mismatch. We have also considered the effect of imperfect state preparations, which lead the input state to have a small component not belong to the Hilbert space where the measurement device satisfies fair sampling. In this case, the presence of off-diagonal coherence terms in the density matrix can make the error in the statistics to scale (unfavourably) as the square root of the state preparation error.

We have also argued that the measurement devices can be securely used in cryptographic protocols, provided that they satisfy fair sampling. It is paramount in this case that fair sampling holds also from the perspective of the adversary. If, for instance, the adversary possesses a more complete description (a.k.a., a purification) of the measurement devices where fair sampling does not hold, explicit attack strategies can be exhibited.

\section*{Acknowledgements}

We thank Matthias Bock and Christoph Becher for discussions which motivated us to work on the effect of post-selections on device-independent claims. We acknowledge  funding  by  the Swiss National Science Foundation (SNSF) through the Grant PP00P2-179109 and by the Army Research  Laboratory  Center  for  Distributed  Quantum Information via the project SciNet.

\bibliographystyle{unsrt}

\newpage

\section*{Appendices}
\renewcommand\thesubsection{\Alph{subsection}}
\setcounter{section}{-1}
\setcounter{subsection}{0}

\subsection{Necessary conditions for fair sampling}
\label{app:necessary}

Fair sampling is formulated referring to ideal lossless devices, so it is not experimentally testable. However, fair sampling does imply restrictions in the correlations that can be observed, and these provide some consistency checks that the experimenter can perform using the classical outcomes alone. Note that fair sampling gives an independent set of equations from those obtained from the \emph{no-signalling} condition.

The local quantum state that the one party measures in a multipartite experiment is
\begin{align}
	\rho \; = \; \Psi_{\vec{\a}_R}^{\vec{\x}_R}
\end{align}
where the indices imply that the state may depend on the measurement settings $\vec{\x}_R$ and outcomes $\vec{\a}_R$ obtained by the other parties, a phenomenon known as \emph{quantum steering}. Fair sampling for the device of a given party requires that the efficiency factorizes in a component dependent on $\x$ and in a component dependent on $\rho = \Psi_{\vec{\a}_R}^{\vec{\x}_R}$, see Proposition~\ref{prop:equivalence}. Hence:
\begin{align}
	\cE(\x, \vec{\x}_R,\vec{\a}_R) 
	\; = \; 
	\cE_C(\x) \; \cE_Q(\vec{\x}_R,\vec{\a}_R) \;.
\end{align}
If impose the \emph{strong} version of fair sampling, the dependence on $\rho$ disappears, see Proposition~\ref{prop:equivalence2}. Hence we have:
\begin{align}
	\cE(\x, \vec{\x}_R,\vec{\a}_R) 
	\; = \; 
	\cE_C(\x) \;.
\end{align}

Obviously, these are not \emph{sufficient} conditions to have fair sampling. Even when all devices have constant efficiency for all parties, $\cE(\x, \vec{\x}_R,\vec{\a}_R) = const$, the detection loophole is still open.

\subsection{Filter diagonal form}
\label{app:filter}

We consider a lossy measurement device $\cM = \ocM \circ \cF$ which satisfies fair sampling, $\cF = \cF_C \otimes \cF_Q$, and we consider the successful branch $\cF_{C,\cm}$ of $\cF_C$, which is a (non normalized) probabilistic map. Here, we show that we can use the freedom in defining the filters to bring $\cF_{C,\cm}$ to a diagonal form $\mathcal{T}_{C,\cm}$ which does not change $\x$.

The general form of the probabilistic filter $\cF_{C,\cm}$ is:
\begin{align}
	\cF_{C,\cm}(\proj{\x}) 
	\; & = \;
	\sum_{\x' \in \X} \Pr(\x'|\,\x) \, \sproj{\x'} \;,
\end{align}
where $\Pr(\x'|\x)$ are sub-normalized probabilities, i.e., $\Pr(\cm|\,\x) = \sum_{\x' \in \X} \Pr(\x'|\,\x) \leq  1$. We define a new filter $\mathcal{T}_{C,\cm}$ and a new device $\bar{\mathcal{W}}$ (with POVM elements $\bar{W}^\x_\a$) as:
\begin{align}
	& \mathcal{T}_{C,\cm}(\proj{\x}) 
	\; := \;
	\Pr(\cm|\,\x) \proj{\x} 
	\\
	& \bar{W}_\a^\x \; := \; 
	\frac{1}{\Pr(\cm|\,\x)}
	\sum_{\x'\in \X}  
	\Pr(\x'|\,\x) \, \bar{M}_\a^{\x'} \;.
\end{align}
We then have $\ocM \circ (\cF_{C,\cm} \otimes \t{Id}_Q) = \bar{\mathcal{W}} \circ (\mathcal{T}_{C,\cm} \otimes \t{Id}_Q)$\footnote{In the computation in Eq.~\eqref{eq:filter_NF} we omit the flag register $\proj{\cm}$ for sake of clarity.}:
\begin{align}
\label{eq:filter_NF}
{\Pr}_{\ocM \circ (\cF_{C,\cm} \otimes \t{Id}_Q)}(\a\, | \,\x, \rho)
	\; & = \;
	\Tr \Big(\,
	{\sum}_{\x'} \,	\Pr(\x'|\,\x) \,
	\bar{M}_\a^{\x'} \rho  
	\,\Big) \nn\\
	\; & = \;
	\Tr \Big(
	\Pr(\cm|\,\x) \; \bar{W}_\a^{\x} \, \rho  
	\,\Big) \nn\\
	& = \;
	{\Pr}_{\bar{\mathcal{W}} \circ (\mathcal{T}_{C,\cm} \otimes \t{Id}_Q)}(\a\, | \,\x, \rho) \;,
\end{align}
and one can immediately verify that the device $\bar{\mathcal{W}}$ has unit efficiency:
\begin{align}
\forall \x \in \X \qquad
	\sum_{\a \in \A} \bar{W}_\a^\x
	\; = \;  
	\sum_{\x' \in \X}  \frac{\Pr(\x'|\,\x)}{\Pr(\cm|\,\x)} \,
	\sum_{\a \in \A} \bar{M}_\a^{\x'} 
	\; = \; \proj{\cm} \otimes \Id \;.
\end{align}

\subsection{Complementary computations for Proposition 2}
\label{app:complementary}

In the following computation we plug the state $\Psi_\cm$ of Eq.~\eqref{eq:Psi'} into the measurement devices $\cM_k^\cm$ defined in Eq.~\eqref{eq:M'}, and verify that Eq.~\eqref{eq:post-selection_equality} holds. 
\begin{align}
\label{C1}
	& {\Pr}_{\overrightarrow{\cM}^\cm}\big(\vec{\a}\,|\,\vec{\x}, \Psi_\cm\big)
	\; = \;
	\big( \, \overrightarrow{\cM}^\cm \big)^{\vec{\x}}_{\vec{\a}} \ (\Psi_\cm)\\[2pt]
\label{C2}
	& = \;
	\frac{1}{\cE_C(\vec{\x}) \, \cE_Q(\Psi)}
	\Tr\Big\lbrace 
	\Big({\bigotimes}_k \bar{M}_k\Big)_{\vec{\a}} \;
	\Big[\, (\proj{\cm})^{\otimes N} 
	\otimes \cF_{C,\cm} (\proj{\vec{\x}}) \otimes 
	\cF_{Q,\cm} (\Psi ) \, 
	\Big]\,\Big\rbrace \\
\label{C3}
	& = \;
	\frac{1}{\cE_C(\vec{\x}) \, \cE_Q(\Psi)}
	\Tr\Big\lbrace 
	\Big({\bigotimes}_k \bar{M}_k\Big)_{\vec{\a}} \;
	\Big[\, (\proj{\cm})^{\otimes N} 
	\otimes \cF_{\cm} 
	(\proj{\vec{\x}} \otimes \Psi) \, 
	\Big]\,\Big\rbrace \\
%
%
\label{C5}
	& = \;
	\frac{1}{\cE_C(\vec{\x}) \, \cE_Q(\Psi)}
	\Big({\bigotimes}_k \ocM_k \circ \cF_{k} \Big)_{\vec{\a}} \;
	(\proj{\vec{\x}} \otimes \Psi) \\
\label{C6}
	& = \;
	\frac{1}{\cE_C(\vec{\x}) \, \cE_Q(\Psi)}
	\Big({\bigotimes}_k \cM_k \Big)_{\vec{\a}} \;
	(\proj{\vec{\x}} \otimes \Psi) \\
\label{C7}
	& = \;
	\frac{{\Pr}_{\overrightarrow{\cM}}(\vec{\a}\,|\,\vec{\x}, \Psi)}
	     {{\Pr}_{\overrightarrow{\cM}}(  \cm     |\,\vec{\x}, \Psi)} \;.
\end{align}
In line \eqref{C2}, we have introduced $\cF_{C,\cm} := \bigotimes_k \cF_{k,C,\cm}$ and $\cF_{C,\cm} := \bigotimes_k \cF_{k,C,\cm}$. In line \eqref{C3}, we defined $\cF_\cm := \cF_{C,\cm} \otimes \cF_{Q,\cm}$. In line \eqref{C5}, we have used the property that each device $\ocM_k$ returns a valid outcome $\a_k$ only when the filter $\cF_k$ returns $\cm$. Finally, in line \eqref{C6} and \eqref{C7} we have used the definitions $\cM_k = \ocM_k \circ \cF_k$ and $\Pr_{ \overrightarrow{\cM} } (\cm|\,\vec{\x}, \Psi) = \cE(\vec{\x},\Psi) = \cE_C(\vec{\x}) \, \cE_Q(\Psi)$.

\subsection{Proof of Proposition 4}
\label{app:proof}

We now give a constructive proof of Proposition~\ref{prop:approximate_fs}, meaning that we provide an explicit construction of a lossless device $\ocM$ for which inequality~\eqref{eq:approx_result} holds. From preliminary numerical investigations we conjecture that the bound is not tight and could be further improved. Moreover, an explicit optimization over $\ocM$ (for a given $\cM$) could provide better bounds.

We introduce some handy notation. Given an operator $M_{Q,\cm} \succcurlyeq 0$, we define:
\begin{align}
\begin{split}
	& \widetilde{M}_\a^\x
	\; := \;
	\sqrt{M_{Q,\cm}^{-1}} \,M_\a^\x\, \sqrt{M_{Q,\cm}^{-1}} \\
	& \widetilde{M}_\cm^\x
	\; := \;
	\sqrt{M_{Q,\cm}^{-1}} \,M_\cm^\x\, \sqrt{M_{Q,\cm}^{-1}} \;,
\end{split}	
\end{align}
so that condition~\eqref{eq:approx} can be compactly written as
\begin{align}
\label{eq:approx2}
	\norm{\, \Pi_\cm - 
	\frac{\widetilde{M}_\cm^\x}{\bnorm{\widetilde{M}_\cm^\x} } 
	\,}
	\;\leq\;
	\epsilon \;,
\end{align}
where we remind that $\Pi_\cm$ is the projector on the support of $M_{Q,\cm}$ and the inverse square roots are defined on its support; then, the support of $\widetilde{M}_\cm^\x$ is contained in the support of $M_{Q,\cm}$. We also require $\Pi_\cm M_\cm^\x \Pi_\cm = M_\cm^\x$; it then follows $\Pi_\cm M_\a^\x \Pi_\cm = M_\a^\x$ for all $\x$ and for all $\a$.

Now we define the filtered state $\rho_\cm$ as
\begin{align}
	\rho_\cm \; := \; 
	\frac{1}{\cE_Q(\rho)} \sqrt{M_{Q,\cm}} \, \rho \, \sqrt{M_{Q,\cm}}
\end{align}
with $\cE_Q(\rho) = \Tr(M_{Q,\cm} \rho)$. Notice that $\rho_\cm$ can be filtered probabilistically from $\rho$ and, moreover, $\Pi_\cm \rho_\cm \Pi_\cm = \rho_\cm$.

The definition of the lossless detector $\ocM$ is more laborious. We first introduce: 
\begin{align}
	\Delta^\x 
	\; := \; 
	\Pi_\cm - 
	\frac{\widetilde{M}_\cm^\x}{\bnorm{\widetilde{M}_\cm^\x}} 
\end{align}
and from inequality~\eqref{eq:approx2} it follows that:
\begin{align}
	0 \preccurlyeq \Delta^\x \preccurlyeq \epsilon \,\Pi_\cm \;.
\end{align}
We can then define $\bar{M}_\a^\x$ as:
\begin{align}
	\bar{M}_\a^\x
	\; := \;
	\frac{\widetilde{M}_\a^\x}{\bnorm{\widetilde{M}_\cm^\x}} +
	\Delta_\a^\x \;,
\end{align}
where $\{\Delta_\a^\x\}$ is an arbitrary set of positive semi-definite operators that satisfy $\sum_{\a\in\A} \Delta_\a^\x = \Delta^\x$ for all $\x$. Explicitly, we may choose $\Delta_\a^\x = \Delta^\x / |\A|$, where $|\A|$ is the size of $\A$, so that we have:
\begin{align}
	\bar{M}_\a^\x
	\; = \;
	\frac{\widetilde{M}_\a^\x}{\bnorm{\widetilde{M}_\cm^\x}} +
	\frac{1}{|\A|}
	\left(
	\Pi_\cm- \frac{\widetilde{M}_\cm^\x}{\bnorm{\widetilde{M}_\cm^\x}}
	\right) 
	\label{eq:M_bar}
\end{align}
and one can immediately verify that $\bar{M}_\a^\x \succcurlyeq 0$ and $\sum_{\a\in\A} \bar{M}_\a^\x = \Pi_\cm$.

With all the definitions in place, we proceed to verify inequality~\eqref{eq:approx_result}. First, notice that:
\begin{align}
	\frac{\Tr(M_\a^\x  \rho)}
	     {\Tr(M_\cm^\x \rho)}
	\; = \;
	\frac{\Tr(\Pi_{\cm}M_\a^\x\, \Pi_{\cm} \rho)}
	     {\Tr(\Pi_{\cm}M_\cm^\x\,\Pi_{\cm} \rho)} 
	\; = \;
	\frac{\Tr\big(\widetilde{M}_\a^\x  \rho_\cm\big)}
	     {\Tr\big(\widetilde{M}_\cm^\x \rho_\cm\big)} \;,
\end{align}
where we have used $\Pi_\cm M_\a^\x \Pi_\cm = M_\a^\x$ and $\Pi_{\cm} = M_{Q,\cm}^{1/2} \, M_{Q,\cm}^{-1/2}$. Then we have:
\begin{align}
	& \hspace{-5mm}
	\sum_{\a\in\A} \frac{1}{2}
	\left\vert\, 
	\frac{\Tr(M_\a^\x \rho)}{\Tr(M_\cm^\x \rho)} -
	\Tr(\bar{M}_\a^\x \rho_\cm)
	\,\right\vert 
	\; = \;
	\sum_{\a\in\A} \frac{1}{2}
	\left\vert\, 
	\frac{\Tr\big(\widetilde{M}_\a^\x  \rho_\cm\big)}
	     {\Tr\big(\widetilde{M}_\cm^\x \rho_\cm\big)}  -
	\Tr(\bar{M}_\a^\x \rho_\cm)
	\,\right\vert 
\label{l1} \\
	= \ &
	\frac{1}{2\,\Tr\big(\widetilde{M}_\cm^\x \rho_\cm\big)} 
	\sum_{\a\in\A} 
	\left\vert\, 
	\Tr\big(\widetilde{M}_\a^\x \rho_\cm\big) -
	\Tr\big(\bar{M}_\a^\x \rho_\cm\big) \Tr\big(\widetilde{M}_\cm^\x \rho_\cm\big)
	\,\right\vert 
\label{l2} \\
	\leq \ &
	\frac{1}{2\,(1-\epsilon) \bnorm{\widetilde{M}_\cm^\x}} 
	\sum_{\a\in\A} 
	\max_{\sigma_\a \in \{0,1\}}
	\left\vert\, 
	\Tr \left\{ \left[
	\widetilde{M}_\a^\x - 
	\left( \frac{\widetilde{M}_\a^\x}{\bnorm{\widetilde{M}_\cm^\x}} + \Delta_\a^\x \right) 
	(1 - \sigma_\a \epsilon) \,\bnorm{\widetilde{M}_\cm^\x}
	\right] \rho_\cm \right\}
	\,\right\vert 
\label{l4} \\	
	= \ &
	\frac{1}{2\,(1-\epsilon) \bnorm{\widetilde{M}_\cm^\x}} 
	\sum_{\a\in\A} 
	\max_{\sigma_\a \in \{0,1\}}
	\left\vert\, 
	\Tr\left\{ \left[\phantom{\Big\vert}
	\sigma_\a \epsilon \, \widetilde{M}_\a^\x - 
	\Delta_\a^\x \, (1 - \sigma_\a \epsilon) \,\bnorm{\widetilde{M}_\cm^\x}  
	\,\right] \rho_\cm\right\}
	\,\right\vert 
\label{l5} \\	
	\leq \ &
	\frac{1}{2\,(1-\epsilon) \bnorm{\widetilde{M}_\cm^\x}}  
	\left[\,
	\sum_{\a\in\A} \epsilon \,\Tr\big(\widetilde{M}_\a^\x \rho_\cm\big) + 
	\sum_{\a\in\A} \bnorm{\widetilde{M}_\cm^\x}\, \Tr\big(\Delta_\a^\x\, \rho_\cm\big) 
	\,\right] 
\label{l6} \\
	= \ &
	\frac{1}{2\,(1-\epsilon) \bnorm{\widetilde{M}_\cm^\x}}  
	\left[ \phantom{\Big\vert}
	\epsilon \,\Tr\big(\widetilde{M}_\cm^\x \rho_\cm\big) + 
	\bnorm{\widetilde{M}_\cm^\x}\, \Tr(\Delta^\x\, \rho_\cm) 
	\,\right] 
\label{l7} \\
	\leq \ &
	\frac{1}{2\,(1-\epsilon) \bnorm{\widetilde{M}_\cm^\x}}  
	\left[ \phantom{\Big\vert}
	\epsilon \, \bnorm{\widetilde{M}_\cm^\x} + 
	\bnorm{\widetilde{M}_\cm^\x} \, \epsilon 
	\,\right] 
\label{l8} \\
	= \ &
	\frac{\epsilon}{1-\epsilon} \;.			
\end{align}
In line \eqref{l4}, we used $(1-\epsilon) \bnorm{\widetilde{M}_\cm^\x} \leq \Tr(\widetilde{M}_\cm^\x \rho_\cm) \leq \bnorm{\widetilde{M}_\cm^\x}$ and the definition of $\bar{M_\a^\x}$. In line~\eqref{l6}, we employed the triangle inequality, then maximized two terms separately over $\sigma_\a \in \{0,1\}$, and also used the fact that $\widetilde{M}_\a^\x$ and $\Delta_\a^\x$ are positive semi-definite to get rid of the absolute values. In line~\eqref{l7}, the definitions of $\widetilde{M}_\cm^\x$ and $\Delta^\x$. Finally, in line~\eqref{l8}, $\Tr(\widetilde{M}_\cm^\x \rho_\cm) \leq \bnorm{\widetilde{M}_\cm^\x}$ and $\norm{\Delta^\x} \leq \epsilon$.

\subsection{Approximate fair sampling in multipartite settings}
\label{app:multipartite}

Suppose that in an $N$-local quantum experiment each party holds a measurement device $\cM_k$ with POVM elements $M^{\x_k}_{k,\cm}$ which approximately satisfies fair sampling, as in Eq.~\eqref{eq:approx}. With the same conventions as in Appendix~\ref{app:proof}, we introduce:
\begin{align}
	\Delta_{k}^{\x_k}
	\; := \;
	\Pi_{k,\cm} - \frac{\widetilde{M}_{k,\cm}^{\x_k}}{\big|\big| \widetilde{M}_{k,\cm}^{\x_k} \big|\big|} 
\end{align}
so that the approximate fair sampling condition is expressed as $\norm{\Delta_{k}^{\x_k}} \leq \epsilon_k$, for some small $\epsilon_k$. Equivalently, we can write:
\begin{align}
	\widetilde{M}_{k,\cm}^{\x_k} 
	\; = \; 
	\big|\big| \widetilde{M}_{k,\cm}^{\x_k} \big|\big| \,
	(\Pi_{k,\cm} - \Delta_k^{\x_k}) \;.
\end{align}
We define also the following operators for a setting $\vec{\x} := (\x_1,\ldots, \x_N)$:
\begin{align}
	\widetilde{M}_\cm^{\vec{\x}}
	\; & := \;
	\widetilde{M}_{1,\cm}^{\x_1} \otimes \cdots \otimes \widetilde{M}_{N,\cm}^{\x_N} \\
	\Pi_\cm^{\vec{\x}}
	\; & := \;
	\Pi_{1,\cm}^{\x_1} 
	\otimes \cdots \otimes
	\Pi_{N,\cm}^{\x_N} 
\end{align}

We can finally define an operator $\Delta^{\vec{\x}}$ which expresses how much the collection of $N$ measurement devices deviates from the fair sampling assumption:
\begin{align}
	\Delta^{\vec{\x}}
	\; := & \;
	\Pi_\cm - \frac{\widetilde{M}_\cm^{\vec{\x}}}{\bnorm{\widetilde{M}_\cm^{\vec{\x}}}} \nn \\
	\; = & \;
	\Pi_\cm - 
	\frac{\widetilde{M}_{1, \cm}^{\x_1}}{\big|\big|{\widetilde{M}_{1, \cm}^{\x_1}}\big|\big| } 
	\otimes \cdots \otimes 
	\frac{\widetilde{M}_{N, \cm}^{\x_N}}{\big|\big|{\widetilde{M}_{N, \cm}^{\x_N}}\big|\big| } \nn \\
	\; = & \;
	\Pi_\cm - 
	(\Pi_{1,\cm} - \Delta_1^{\x_1}) \otimes \cdots \otimes (\Pi_{N,\cm} - \Delta_N^{\x_N}) \;.	
\end{align}
From this follows the bound:
\begin{align}
	\big|\big|{\Delta^{\vec{\x}}}\big|\big|
	\; & \leq \;
	1 - (1-\epsilon_1) \cdots (1-\epsilon_N)  
	\; =: \;
	\epsilon_\t{tot} \;.
\end{align}
which is the approximate strong fair sampling condition given in Eq.~\eqref{eq:approx}, where the approximation error is $\epsilon_\t{tot}$.
Assuming that all the $\epsilon_k$ are small we have:
\begin{align}
\label{eq:epsilon_tot}
	\epsilon_\t{tot}
	\; & = \;
	\epsilon_{1} + \ldots + \epsilon_{N} \; + \;
	O \bigg(	\Big( {\sum}_k \, \epsilon_k \Big)^{\!2} \bigg) \;.
\end{align}

This means that Proposition~\ref{prop:approximate_fs} can be directly applied to a multipartite setting, and the error is upper bounded by $\epsilon_\t{tot}$, which is roughly equal to the sum of the deviations from exact fair sampling for each one of the individual measurement devices.

\subsection{Approximate fair sampling for Bell operators}
\label{app:X}

Consider $N$ spatially separated and non-communicating measurement devices that are used to measure the expectation value of a Bell operator $B$. This expectation value is entirely determined by the outcome probabilities of the measurement devices and for linear Bell inequalities it can be written as:
\begin{align}
	\langle\, B \,\rangle_{\Pr(\cdot|\cdot)}
	\; = \;
	\sum_{{\vec{\a}}, \vec{\x}} 
	c_{\vec{\a}}^{\vec{\x}} \;
	\Pr(\vec{\a}\,|\,\vec{\x}) \;,
\end{align}
for some coefficients $c_{\vec{\a}}^{\vec{\x}} \in \mathds{R}$. Two sets of probability distributions are of interest in our case: the post-selected probability distributions $\Pr_\t{p.s.}$ obtained by the device ``in the lab'', where we consider a ``good outcome'' one where all the $N$ detectors click in a given experimental round; and the probability distributions $\Pr_\t{ideal}$ that would be obtained with lossless devices. The devices satisfy an approximate fair sampling condition.

We want to bound the quantity
\begin{align}
	\frac{1}{\beta_\t{max}}\;
	\big|\, \langle\, B \,\rangle_\t{p.s.} - \langle\, B \,\rangle_\t{ideal} \,\big|
\end{align}
where $\beta_\t{max}$ is the algebraic bound for the Bell operator $B$, and is introduced here to compensate for arbitrary rescaling of the coefficients $	c_{\vec{\a}}^{\vec{\x}}$. The algebraic bound is defined as:
\begin{align}
	\beta_\t{max} 
	\; := & \;
	\max_{\{\Pr(\cdot|\cdot)\}}
	\left\vert
	\sum_{{\vec{\a}}, \vec{\x}} 
	c_{\vec{\a}}^{\vec{\x}} \;
	\Pr(\vec{\a}\,|\,\vec{\x})
	\right\vert 
\end{align}
where the maximum is over the set of probability distributions $\Pr(\vec{\a}|\vec{\x})$. We also suppose that the total variation distance between $\Pr_\t{p.s.}(\cdot|\vec{\x})$ and $\Pr_\t{ideal}(\cdot|\vec{\x})$ is bounded by $\epsilon_\t{tot}$ for all collections of settings $\vec{\x}$:
\begin{align}
	\norm{\;
	{\Pr}_\t{p.s.}(\,\cdot\,|\,\vec{\x}) - {\Pr}_\t{ideal}(\,\cdot\,|\,\vec{\x})
	\;}_{TV}	
	\; \leq \;
	\epsilon_\t{tot} \;.
\end{align}
Then we have for all $\vec{\x}$:
\begin{align}
	{\Pr}_\t{p.s.}(\,\cdot\,|\,\vec{\x}) - {\Pr}_\t{ideal}(\,\cdot\,|\,\vec{\x})
	\; = \;
	\mu_{\vec{\x}} \; q_1(\,\cdot\, |\, \vec{\x}) - 
	\mu_{\vec{\x}} \; q_2(\,\cdot\, |\, \vec{\x})
\end{align}
where $q_1(\cdot |\vec{\x})$ and $ q_2(\cdot |\vec{\x})$ are probability distributions having disjoint support (i.e., $q_1(\cdot |\vec{\x})= 0$ if $q_2(\cdot |\vec{\x})\neq 0$ and vice-versa) and $\mu_{\vec{\x}}\in [0, \epsilon_\t{tot}]$ is a re-scaling constant. We can then compute:
\begin{align}
	\frac{1}{\beta_\t{max}}\;
	\big|\, \langle\, B \,\rangle_\t{p.s.} - \langle\, B \,\rangle_\t{ideal} \,\big|
	\; & = \;
	\frac{1}{\beta_\t{max}}
	\left\vert
	\sum_{{\vec{\a}}, \vec{\x}} 
	c_{\vec{\a}}^{\vec{\x}} \;
	\Big[
	{\Pr}_\t{p.s.}(\vec{\a}\,|\,\vec{\x}) - {\Pr}_\t{ideal}(\vec{\a}\,|\,\vec{\x})
	\Big]
	\right\vert \\
	\; & = \;
	\frac{1}{\beta_\t{max}}
	\left\vert
	\sum_{{\vec{\a}}, \vec{\x}} 
	c_{\vec{\a}}^{\vec{\x}} \, \mu_{\vec{\x}} \;
	\Big[
	 q_1(\vec{\a}\, |\, \vec{\x}) - 
	q_2(\vec{\a}\, |\, \vec{\x})
	\Big]
	\right\vert \\	
\label{F1}
	\; & \leq \;
	\frac{1}{\beta_\t{max}}
	\left[\;
	\epsilon_\t{tot}
	\left\vert
	\sum_{{\vec{\a}}, \vec{\x}} 
	c_{\vec{\a}}^{\vec{\x}} \;
	q_1(\vec{\a}\, |\, \vec{\x}) 
	\right\vert +
	\epsilon_\t{tot}
	\left\vert
	\sum_{{\vec{\a}}, \vec{\x}} 
	c_{\vec{\a}}^{\vec{\x}} \;
	q_1(\vec{\a}\, |\, \vec{\x}) 
	\right\vert
	\;\right] \\
\label{F2}
	& \leq \;
	2 \, \epsilon_\t{tot} \;.
\end{align}
In line~\eqref{F1} we have used the triangle inequality and in line~\eqref{F2} the definition of $\beta_\t{max}$.

\subsection{Computation for approximate fair sampling with a polarization analyser}
\label{app:F}

We want to show that the POVM $M_\cm^\theta$ defined in Eq.~\eqref{eq:multi-photon_approx} in Section~\ref{sec:appr_fs_in_quantum_optics} satisfies approximate fair sampling when employing the operator $M_{Q,\cm} = \Id - R^{\hat{N}}$. We start computing $\big|\big|\widetilde{M}_\cm^\theta \big|\big|$, where $\widetilde{M}_\cm^\theta = (M_{Q,\cm})^{-1/2}\, M_\cm^\theta\,(M_{Q,\cm})^{-1/2}$. As usual, the inverse square roots are defined on the support of the operator $M_{Q,\cm}$, in this case on the subspace orthogonal to $\ket{\vac}$. We use $[\hat{N}, \hat{N}_\theta] = 0$, hence we have:
\begin{align}
	\widetilde{M}_\cm^\theta 
	\; & = \;
	(M_{Q,\cm})^{-1}\, M_\cm^\theta \nn\\
	& = \;
	\big(\Id - R^{\hat{N}}\big)^{-1} \,
	\big[\Id - R^{\hat{N}} (1+\delta)^{\hat{N}_\theta}\big] \;.
\end{align}
From $(1+\delta)^{\hat{N}_\theta} \succcurlyeq \Id$ we have $[\Id - R^{\hat{N}} (1+\delta)^{\hat{N}_\theta}] \preccurlyeq (\Id - R^{\hat{N}})$ and thus $\widetilde{M}_\cm^\theta \preccurlyeq \Id$, equivalently, $\big|\big| \widetilde{M}_\cm^\theta \big|\big| \leq 1$. This upper bound is saturated, e.g., by the state $\sket{\theta^\bot} = b_{\theta^\bot}^\dag \ket{\vac}$, since $\sbra{\theta^\bot}\widetilde{M}_\cm^\theta \sket{\theta^\bot} =  (1-R)^{-1} (1-R) = 1$, and thue we conclude $\big|\big| \widetilde{M}_\cm^\theta \big|\big| = 1$.

Then we can calculate
\begin{align}
\label{m1}
	\norm{\,
	\Pi_\cm -
	\frac{\widetilde{M}_\cm^\theta }
	{\big|\big|{\widetilde{M}_\cm^\theta}\big|\big|}
	\,}
	\; & = \; 
	\norm{\,
	\Pi_\cm -
	\big(\Id - R^{\hat{N}}\big)^{-1} \,
	\big[\Id - R^{\hat{N}} (1+\delta)^{\hat{N}_\theta}\big]
	\,} \\[-1mm]
\label{m2}
	\; & = \; 
	\norm{\,
	\big(\Id - R^{\hat{N}}\big)^{-1} \,
	\big[\Id - R^{\hat{N}} - \Id + R^{\hat{N}} (1+\delta)^{\hat{N}_\theta}\big]
	\,} \\[1mm]
\label{m3}
	\; & = \; 
	\norm{\,
	\big(\Id - R^{\hat{N}}\big)^{-1} \,
	\big[ R^{\hat{N}} (1+\delta)^{\hat{N}_\theta} - R^{\hat{N}}\big]
	\,} \\
\label{m4}
	\; & = \;
	\max_{n \geq 1} \
	\frac{R_1^n - R_2^n}{1 - R_2^n} \;.
\end{align}
In line~\eqref{m2} we have used $\Pi_\cm = \big(\Id - R^{\hat{N}}\big)^{-1} \, \big(\Id - R^{\hat{N}}\big)$. To obtain line~\eqref{m4} we have used the fact that the operator in line~\eqref{m3} commutes with $\hat{N}$, thus it can be written as a direct sum of operators having a definite number of photons; within an $n$-photon subspace, the operator $(\Id - R^{\hat{N}})^{-1}$ takes value $1/(1-R)^n = 1/(1-R_2)^n$, while the operator $[ R^{\hat{N}} (1+\delta)^{\hat{N}_\theta} - R^{\hat{N}}]$ attains the maximum value $[R \,(1+\delta)]^n - R^n = R_1^n -R_2^n$ when all $n$ photons are polarized in the $\theta$ direction. Maximizing over all $n$-photon subspaces (with $n \neq 0$) one obtains the value of the operator norm. But the maximum is always attained for $n=1$, since:
\begin{align}
	\frac{R_1^n - R_2^n}{1 - R_2^n} 
	\; = \;
	\frac{R_1 - R_2}{1 - R_2} \;
	\frac{ R_1^{n-1} +  R_1^{n-2} R_2 + ~\cdots~ + R_2^{n-1}}
	{\hfill 1~+\phantom{R_1^{n-2}}R_2 + ~\cdots~ + R_2^{n-1}}
	\; \leq \;
	\frac{R_1 - R_2}{1 - R_2} \;.	
\end{align}
This finally results in:
\begin{align}
	\norm{\,
	\Pi_\cm -
	\frac{\widetilde{M}_\cm^\theta }
	{\big|\big|{\widetilde{M}_\cm^\theta}\big|\big|}
	\,}
	\; & = \; 
	\frac{R_1 - R_2}{1 - R_2} 
	\; = \;
	\frac{R \, \delta}{1 - R}
	\; = \;
	\frac{1 - \eta }{\eta}\, \delta 
 	\;.	
\end{align}

\subsection{Proof of Proposition 5}
\label{app:proof2}

Consider a state $\hat\rho \in S(\hil \oplus \hil^\bot)$ as in Eq.~\eqref{eq:rho_hat}: $\Pi\, \hat \rho \,\Pi = (1-\epsilon') \rho$, $\Pi^\bot \hat \rho \,\Pi^\bot = \epsilon' \rho^\bot$, and $\Pi\, \hat \rho \,\Pi^\bot = c$, where $\Pi, \Pi^\bot$ are the projectors on the well-behaved subspace $\hil$ and on the orthogonal subspace, respectively. Preliminary, we establish the following Lemma.

\begin{lemma}
\label{lemma}
The coherence term satisfies $\norm{\,c\,}_\Tr \leq \sqrt{\epsilon' (1-\epsilon')}$.
\end{lemma}

\begin{proof}
We write $\hat\rho$ as a convex mixture of pure states, $\hat\rho = \sum_i p_i \sproj{\hat\psi_i}$, and define 
\begin{align}
	\sqrt{1-\epsilon_i} \ket{\psi_i} 
	\; := \; 
	\Pi \, \sket{\hat\psi_i} 
	\qquad
	\sqrt{\epsilon_i} \ket{\psi_i^\bot} 
	\; := \; 
	\Pi^\bot \sket{\hat\psi_i} \;.
\end{align}
We then have $\sum_i p_i \epsilon_i = \epsilon'$, $\sum_i p_i (1 -\epsilon_i) = 1- \epsilon'$ and $c = \sum_i p_i \sqrt{\epsilon_i(1-\epsilon_i)} \sket{\psi_i}\!\sbra{\psi_i^\bot}$. Using the triangle and Cauchy-Schwarz inequalities, we obtain the required result:
\begin{align}
	\norm{\,c\,}_\Tr 
	\; & \leq \;
	\sum_i p_i \sqrt{\epsilon_i(1-\epsilon_i)}
	\;\norm{\,\sket{\psi_i}\!\sbra{\psi_i^\bot}\,}_\Tr 
	\\
	\; & = \;
	\sum_i \sqrt{p_i \phantom{\vert\!}\epsilon_i} \sqrt{p_i \, (1-\epsilon_i)} 
	\\
	\; & \leq \;
	\sqrt{\sum_i p_i \, \epsilon_i } \
	\sqrt{\sum_i p_i \, (1-\epsilon_i) } 
	\\
	\; & = \; \sqrt{\epsilon'} \sqrt{1-\epsilon'}
	\;.
\end{align}
\end{proof}

Now, we prove the following inequality:
\begin{align}
\label{G0}
	\sum_{\a \in \A}
	\left\vert\,
	\Tr(\hat{M}_\a^\x \hat{\rho}) - \Tr(M_\a^\x \rho)
	\,\right\vert
	\; & \leq \;
	2 \norm{c}_\Tr \, + \, 2\epsilon' \;.
\end{align}
From $\rho = \Pi \rho \Pi$ and $M_\a^\x = \Pi \hat{M}_\a^\x \Pi$ it follows that $\Tr(M_\a^\x \rho) = \Tr(\hat{M}_\a^\x \rho)$. Thus, we have:
\begin{align}
	\sum_{\a \in \A}
	\left\vert\,
	\Tr(\hat{M}_\a^\x \hat{\rho}) - \Tr(M_\a^\x \rho)
	\,\right\vert
	\; & = \;
	\sum_{\a \in \A}
	\left\vert\,
	\Tr\big[\,\hat{M}_\a^\x (\hat{\rho} - \rho) \,\big]
	\,\right\vert 
	\\
\label{G1}
	\; & = \;
	\sum_{\a \in \A^+}
	\Tr\big[\,\hat{M}_\a^\x (\hat{\rho} - \rho) \,\big]
	+
	\sum_{\a \in \A^-}
	\Tr\big[\,\hat{M}_\a^\x (\rho - \hat{\rho}) \,\big]
	\\
\label{G2}
	\; & = \;
	\Tr\big[\,(\hat{M}_+^\x - M_-^\x) (\hat{\rho} - \rho) \,\big]
	\\[2mm]
	\; & \leq \;
	\norm{\,\hat{\rho} - \rho\,}_\Tr
\end{align}	
In line~\eqref{G1} we split the sum over subsets $\A^+$ and $\A^-$ where $\Tr[\hat{M}_\a^\x (\hat{\rho} - \rho)]$ is, respectively, positive and negative, and in line~\eqref{G2} we defined $\hat{M}_\pm^\x := \sum_{\a \in \A^\pm} \hat{M}_\a^\x$. We then used the operational characterization of the trace norm, $\norm{\sigma}_\Tr = \max_X \Tr(X\sigma)$, where the maximum is taken over hermitian matrices $X$ with $-\Id \preccurlyeq X \preccurlyeq \Id$. We then have:
\begin{align}
\label{G3}
	\norm{\,\hat{\rho} - \rho\,}_\Tr
	\; & = \;
	\norm{\,
	\epsilon' \rho^\bot - \epsilon' \rho + c + c^\dag 
	\,}_\Tr 
	\\[1mm]
\label{G4}
	\; & \leq \;
	2 \norm{\,c \,}_\Tr \, + \,
	\epsilon' \norm{\, \rho^\bot - \rho \,}_\Tr
	\\[1mm]
\label{G5}
	\; & \leq \;
	2 \norm{\,c \,}_\Tr \, + \, 2 \,\epsilon' \;.
\end{align}

Next, using the shorthand $f_\a = \Tr(\hat{M}_\a^\x \hat{\rho})$, $g_\a = \Tr(M_\a^\x \rho)$, $f = \sum_{\a} f_\a$, and $g = \sum_{\a} g_\a$, we have:
\begin{align}
	\left\vert\,
	f - g
	\,\right\vert
	\; \leq \; 
	\sum_{\a \in \A}
	\left\vert\,
	f_\a - g_\a
	\,\right\vert
	\; \leq \;
	2 \norm{\,c\,}_\Tr + 2\,\epsilon'
\end{align}

Finally, we can prove inequality~\eqref{eq:result2}:
\begin{align}
	\sum_{\a\in\A} \frac{1}{2}
	\left\vert\, 
	\frac{\Tr(\hat{M}_\a^\x \hat\rho)}{\Tr(\hat{M}_\cm^\x \hat\rho)} -
	\frac{\Tr(M_\a^\x \rho)}{\Tr(M_\cm^\x \rho)}
	\,\right\vert
	\; & = \;
	\sum_{\a\in\A}
	\frac{1}{2}
	\left\vert\,
	\frac{f_\a}{f} - \frac{g_\a}{g}
	\,\right\vert
	\\
	\; & = \;
	\sum_{\a\in\A}
	\frac{1}{2f}
	\left\vert\,
	f_\a - g_\a + g_\a - g_\a \frac{f}{g}
	\,\right\vert 
	\\
	\; & \leq \;	
	\sum_{\a\in\A}
	\frac{1}{2f}
	\left\vert\,
	f_\a - g_\a
	\,\right\vert +
	\sum_{\a\in\A}
	\frac{g_\a}{2f}
	\left\vert\,
	1 - \frac{f}{g}
	\,\right\vert
	\\
	\; & \leq \;
	\sum_{\a\in\A}
	\frac{1}{2f}
	\left\vert\,
	f_\a - g_\a
	\,\right\vert +
	\frac{1}{2f}
	\left\vert\,
	g - f
	\,\right\vert 
	\\
	\; & \leq \;
	\frac{1}{f} 
	\left(
	2\norm{\,c\,}_\Tr + 2\,\epsilon'
	\right) \;.
\end{align}
By symmetry under exchange of $f_\a$ and $g_\a$, we also have $\frac{1}{g}(2\norm{\,c\,}_\Tr + 2\,\epsilon')$ as an upper bound. Recognizing that $f = \Tr(\hat{M}_\cm^\x\hat\rho)$ and $g = \Tr(M_\cm^\x \rho)$, the proof is concluded.

\end{document}